\documentclass[11pt,onecolumn]{article}

\usepackage{bm}
\usepackage{amsfonts}
\usepackage{amsmath}
\usepackage{amssymb}
\usepackage{cite}
\usepackage[amsmath,thmmarks]{ntheorem}

\usepackage{color, soul}
\usepackage{epic, eepic}
\usepackage[top=1in, bottom=1in, left=1.25in, right=1.25in]{geometry}
\usepackage{appendix}
\usepackage{footnote}
\usepackage{booktabs}
\usepackage[pdftex]{graphicx}

\usepackage{theorem}

\newtheorem{theorem}{Theorem}
\newtheorem{lemma}{Lemma}
\newtheorem{definition}{Definition}

\newtheorem{proposition}{Proposition}
\newtheorem{corollary}{Corollary}

\theoremheaderfont{\sc}\theorembodyfont{\upshape}
\theoremstyle{nonumberplain}
\theoremseparator{}
\theoremsymbol{\rule{1ex}{1ex}}
\newtheorem{proof}{Proof}

\begin{document}

\title{The Convergence Guarantees of a Non-convex Approach for Sparse Recovery}

\author{Laming~Chen and~Yuantao~Gu\thanks{This work was supported by National 973 Program of China (Grant No. 2013CB329201),
National Natural Science Foundation of China (NSFC 61371137, 60872087),
and the autonomous project of science and technology of Tsinghua University with No. of 2012THZ07123. The authors are with State Key Laboratory on Microwave and Digital Communications,
Tsinghua National Laboratory for Information Science and Technology, Department of Electronic Engineering, Tsinghua University, Beijing 100084, China (E-mail: gyt@tsinghua.edu.cn).}}

\date{Received December 24, 2012, revised November 18, 2013.}

\maketitle

\begin{abstract}
In the area of sparse recovery, numerous researches hint that non-convex penalties might induce better sparsity than convex ones, but up until now those corresponding non-convex algorithms lack convergence guarantees from the initial solution to the global optimum. This paper aims to provide performance guarantees of a non-convex approach for sparse recovery. Specifically, the concept of weak convexity is incorporated into a class of sparsity-inducing penalties to characterize the non-convexity. Borrowing the idea of the projected subgradient method, an algorithm is proposed to solve the non-convex optimization problem. In addition, a uniform approximate projection is adopted in the projection step to make this algorithm computationally tractable for large scale problems. The convergence analysis is provided in the noisy scenario. It is shown that if the non-convexity of the penalty is below a threshold (which is in inverse proportion to the distance between the initial solution and the sparse signal), the recovered solution has recovery error linear in both the step size and the noise term. Numerical simulations are implemented to test the performance of the proposed approach and verify the theoretical analysis.

\end{abstract}

\section{Introduction}\label{Sec_Int}

Since the introduction of compressive sensing (CS) \cite{Candes,Donoho,Tao}, sparse recovery has received much attention and becomes a very hot topic these years \cite{Elad,MRI,Radar,BADC,EADC}. Sparse recovery aims to solve the following underdetermined linear system
\begin{align}\label{sparserecovery}
{\bf y}={\bf Ax},
\end{align}
where ${\bf y}\in\mathbb{R}^M$ denotes the measurement vector, ${\bf A}\in\mathbb{R}^{M\times N}$ is a sensing matrix with more columns than rows, i.e., $M<N$, and ${\bf x}=(x_i)\in\mathbb{R}^N$ is the sparse or compressible signal to be recovered.

Many algorithms have been proposed to solve the problem (\ref{sparserecovery}). If $\bf x$ is sparse, one typical method is to consider the following $\ell_0$-minimization problem
\begin{align}\label{l0opt}
\underset{\bf x}{\operatorname{argmin}}\|{\bf x}\|_0\ \ \textrm{subject to}\ \ {\bf y}={\bf Ax},
\end{align}
where the $\ell_0$ ``norm'' $\|{\bf x}\|_0=\#\{i:x_i\neq0\}$ counts the nonzero elements of $\bf x$. However, it is not practical to adopt this method since it is usually solved by combinatorial search, which is NP-hard. An alternate method \cite{BP} is to replace the $\ell_0$ ``norm'' with the $\ell_1$ norm, i.e.,
\begin{align}\label{l1opt}
\underset{\bf x}{\operatorname{argmin}}\|{\bf x}\|_1\ \ \textrm{subject to}\ \ {\bf y}={\bf Ax}.
\end{align}
The convex $\ell_1$-minimization problem (\ref{l1opt}) is also known as \emph{basis pursuit} (BP). It is certified that under some certain conditions \cite{BPRIP}, the optimal solution of $\ell_1$-minimization is identical to that of $\ell_0$-minimization. This conclusion greatly reduces the computational complexity, since $\ell_1$-minimization can be reformulated as a linear program (LP), and be solved by numerous efficient algorithms \cite{ConvexOpt}.

Another family of sparse recovery algorithms is put forward based on non-convex optimization
\begin{align}\label{lFopt}
\underset{\bf x}{\operatorname{argmin}} J({\bf x})\ \ \textrm{subject to}\ \ {\bf y}={\bf Ax},
\end{align}
where $J(\cdot)$ is a sparsity-inducing penalty. The optimization problem (\ref{lFopt}) is also termed as $J$-minimization \cite{Aldroubi}. These algorithms include focal underdetermined system solver (FOCUSS) \cite{FOCUSS}, iteratively reweighted least squares (IRLS) \cite{IRLS}, reweighted $\ell_1$-minimization \cite{RL1}, smoothed $\ell_0$ (SL0) \cite{SL0}, difference of convex (DC) algorithm \cite{DCAlgorithm}, improved smoothed $\ell_0$ (ISL0) \cite{ISL0}, and zero-point attracting projection (ZAP) \cite{ZAP}. It is theoretically proved \cite{Gribonval,nonconvexopt,lqmin,Saab} and experimentally verified \cite{FOCUSS,IRLS,RL1,SL0,DCAlgorithm,ISL0,ZAP,nonconvexopt,lqmin,Saab} that for some certain non-convex penalties, $J$-minimization tends to derive the sparse solution under weaker conditions than $\ell_1$-minimization. However, the inherent deficiency of multiple local minima in non-convex optimization limits its practical usage, where improper initial criteria might cause the solution trapped into the wrong ones.

The convergence performance of some non-convex sparse recovery algorithms has been studied in literatures. For example, in \cite{IRLSana}, a local convergence result of IRLS \cite{IRLS} for $\ell_p$-minimization with $p\in(0,1)$ is established where the convergence is guaranteed in a sufficiently small neighborhood of the sparse signal. Whether or not this neighborhood contains the initial solution is not discussed. In \cite{3MG}, the majorize-minimize (MM) subspace algorithm is proposed to solve the $\ell_2-\ell_0$ regularized problem and its convergence performance is also provided. Under some certain conditions, it is shown that the generated sequence will converge to a critical point, which is not, however, proved to be the global optimum. In \cite{Mohimani}, the convergence performance of SL0 \cite{SL0} is given. This is done due to the ``local convexity'' of the penalties, and SL0 needs to solve a sequence of optimization problems rather than a single $J$-minimization problem to guarantee convergence to the sparse signal.

This paper aims to provide theoretical convergence guarantees of a non-convex approach for sparse recovery from the initial solution to the global optimum. The question, which naturally appears and mainly motivates this paper, is raised as follows.
\begin{quote}
\emph{Does there exist a computationally tractable} algorithm that guarantees to find the sparse solution to $J$-minimization? If yes, in what circumstances does this statement hold?
\end{quote}
In this paper, exploiting the concept of \emph{weak convexity} \cite{rhoconvex} to characterize the non-convexity of the penalties, the mentioned question is replied as follows.
\begin{quote}
\emph{A computationally tractable non-convex approach is proposed with guarantees that it converges to the sparse solution provided that the non-convexity of the penalty is below a threshold.}
\end{quote}

This paper is organized as follows. Section~\ref{Sec_Pre} introduces the preliminaries of this paper, including the projected subgradient method, the concepts of sparseness measure and weak convexity, and some related state of the art researches. In Section~\ref{Sec_Main}, the main contributions of this paper, including the non-convex approach for sparse recovery and its performance guarantees, are demonstrated. The theoretical analysis and some further discussions are provided in Section~\ref{Sec_Theo}. Numerical simulations are implemented in Section~\ref{Sec_Simu} to verify the theoretical results. All of the proofs are included in Section~\ref{Sec_Proof}, and this paper is concluded in Section~\ref{Sec_Conc}.

\section{Preliminary}\label{Sec_Pre}

For constrained convex optimization problem, the projected subgradient method \cite{PSM} is an algorithm which is very simple to implement and easy to analyze. Specifically, consider the convex optimization
\begin{align}
\underset{\bf x}{\operatorname{argmin}}f({\bf x})\ \ \textrm{subject to}\ \ {\bf x}\in\mathcal{C},
\end{align}
where $f:\mathbb{R}^N\rightarrow\mathbb{R}$ is convex (and possibly nondifferentiable) and $\mathcal{C}\subset\mathbb{R}^N$ is a convex set. Denote $P_{\mathcal{C}}(\cdot)$ as the Euclidean projection on $\mathcal{C}$. The projected subgradient method is given by
\begin{align}
{\bf x}(n+1)=P_{\mathcal{C}}\left({\bf x}(n)-\kappa(n)g(n)\right),
\end{align}
where $\kappa(n)$ and $g(n)$ are the $n$th step size and any subgradient of $f(\cdot)$ at ${\bf x}(n)$, respectively. Theoretical analysis \cite{nonlinear,Boyd_Subgradient} reveals that this method converges to the optimum for some certain types of step size rules, e.g. the step size sequence which is square summable but not summable.

Several notable differences between the projected subgradient method and the ordinary projected gradient method \cite{PGM} should be pointed out. First, the projected subgradient method applies directly to nondifferentiable convex functions while the latter doesn't. Second, the function value of the solution sequence can increase in the projected subgradient method. Therefore, the key quantity is the Euclidean distance to the optimum instead of the function value. In addition, the projected subgradient method adopts step size sequence fixed in advance rather than an exact or approximate line search as in the projected gradient method.

For non-convex $J$-minimization problem (\ref{lFopt}), the projected subgradient method is no longer applicable. The following two subsections introduce the concepts of sparseness measure and weak convexity, by which the projected subgradient method can be generalized to be applicable to $J$-minimization.

\subsection{Sparseness Measure}

First, a class of sparsity-inducing penalties is introduced. The penalty $J({\bf x})$ in (\ref{lFopt}) is defined as
\begin{align}\label{costfunction}
J({\bf x})=\sum_{i=1}^{N} F(x_i),
\end{align}
where $F(\cdot)$ belongs to a class of sparseness measures \cite{Gribonval} satisfying the following Definition~\ref{definition_spar}.

\begin{definition}\label{definition_spar}
The sparseness measure $F:\mathbb{R}\rightarrow\mathbb{R}$ satisfies
\begin{enumerate}
\item
$F(0)=0$, $F(\cdot)$ is even and not identically zero;
\item
$F(\cdot)$ is non-decreasing on $[0,+\infty)$;
\item
The function $t\mapsto F(t)/t$ is non-increasing on $(0,+\infty)$.
\end{enumerate}
\end{definition}

As has been revealed in \cite{Gribonval}, the null space property with its constant \cite{NSP} is closely related to whether $J$-minimization is able to find the sparse signal. Define ${\bf x}_S$ as the vector generated by setting the entries of $\bf x$ indexed by $S^c=\{1,2,\ldots,N\}\setminus S$ to zeros.

\begin{definition}\label{definition_NSP}
Define null space constant $\gamma(J,{\bf A},K)$ as the smallest quantity such that
\begin{align}\label{NSC}
J({\bf z}_S)\le\gamma(J,{\bf A},K) J({\bf z}_{S^c})
\end{align}
holds for any set $S\subset\{1,2,\ldots,N\}$ with $\#S\le K$ and for any vector ${\bf z}\in \mathcal{N}({\bf A})$, where $\mathcal{N}({\bf A})$ denotes the null space of ${\bf A}$.
\end{definition}

Based on Definition~\ref{definition_spar} and Definition~\ref{definition_NSP}, the following proposition is derived in \cite{Gribonval}.

\begin{proposition}\label{proposition_NSP}
(Theorem 2, 3, and 5 from \cite{Gribonval}). For penalty $J(\cdot)$ formed by $F(\cdot)$ satisfying Definition~\ref{definition_spar}, the following statements hold:
\begin{enumerate}
\item
If $\gamma(J,{\bf A},K)<1$, then for any $\bf x$ satisfying $\|{\bf x}\|_0\le K$ and ${\bf y=Ax}$, $\bf x$ is the unique solution to (\ref{lFopt});
\item
If $\gamma(J,{\bf A},K)>1$, then there exist vectors $\bf x$ and $\bf x'$ such that $\|{\bf x}\|_0\le K$, ${\bf Ax=Ax'}$ and $J({\bf x'})<J({\bf x})$;
\item
$\gamma(\ell_0,{\bf A},K)\le\gamma(J,{\bf A},K)\le\gamma(\ell_1,{\bf A},K)$.
\end{enumerate}
\end{proposition}

Proposition~\ref{proposition_NSP}.1)-2) reveals that the null space constant is a tight quantity for the tuple $(J,{\bf A},K)$ to indicate the performance of $J$-minimization. Here the tightness is in the sense that $\gamma(J,{\bf A},K)<1$ implies all $K$-sparse signals are the unique solutions to $J$-minimization, while not all $K$-sparse signals satisfy this if $\gamma(J,{\bf A},K)>1$. Proposition~\ref{proposition_NSP}.3) indicates that for the tuple $({\bf A},K)$, if all $K$-sparse signals are the unique solutions to $\ell_1$-minimization, i.e. $\gamma(\ell_1,{\bf A},K)<1$, this also applies to $J$-minimization. Therefore, \emph{in the worst case sense} which takes over all $K$-sparse signals, the performance of $J$-minimization is at least as good as that of $\ell_1$-minimization.

\subsection{Weak Convexity}

The concept of weak convexity was proposed decades ago \cite{Janin}. A real valued function $F(\cdot)$ defined on a convex subset $S\subseteq\mathbb{R}$ is $\rho$-convex if there exists some real number $\rho$ which is the largest quantity such that the inequality
\begin{align*}
F(\lambda t_1\!+\!(1\!-\!\lambda)t_2)\!\le\!\lambda F(t_1)\!+\!(1\!-\!\lambda)F(t_2)\!-\!\rho\lambda(1\!-\!\lambda)(t_1\!-\!t_2)^2
\end{align*}
holds for any $t_1,t_2\in S$ and for any $\lambda\in[0,1]$. $\rho>0$, $\rho=0$ and $\rho<0$ correspond to strong convexity, convexity and weak convexity, respectively. The following proposition reveals that $F(\cdot)$ can be decomposed into the sum of a convex function and a square.

\begin{proposition}\label{proposition_1}
(Proposition 4.3 from \cite{rhoconvex}). Function $F:S\rightarrow\mathbb{R}$ is $\rho$-convex if and only if there exists a convex function $H:S\rightarrow\mathbb{R}$ such that $F(t)=H(t)+\rho t^2$ for all $t\in S$.
\end{proposition}

According to Proposition~\ref{proposition_1}, weakly convex functions are also known as semi-convex functions \cite{Colesanti}. For any $t\in{\rm int}S$ which denotes the interior of $S$, define the directional derivative of a $\rho$-convex function $F(\cdot)$ as
\begin{align}
D_F(t;\nu)=\lim_{\theta\rightarrow0_+}\frac{F(t+\theta\nu)-F(t)}{\theta},
\end{align}
then the generalized gradient set \cite{Generalized} is defined as
\begin{align}\label{definition_gen}
\partial F(t)=\{f(t):\nu f(t)\le D_F(t;\nu),\ \forall\nu\in\mathbb{R}\}.
\end{align}
If $F(\cdot)$ is convex, $\partial F(\cdot)$ is commonly known as the subgradient set. The following proposition demonstrates an important property of $\rho$-convex functions which will be used in the theoretical analysis.

\begin{proposition}\label{proposition_3}
(Proposition 4.8 from \cite{rhoconvex}). Let $F(\cdot)$ be $\rho$-convex on $S$, then for any $t_1\in {\rm int}S$, $t_2\in S$, and for any $f(t_1)\in\partial F(t_1)$,
\begin{align}
F(t_2)\ge F(t_1)+f(t_1)(t_2-t_1)+\rho(t_2-t_1)^2.
\end{align}
\end{proposition}

\subsection{Related Work}

Before formally introducing the main results of our paper, some related state-of-the-art researches are introduced. Being aware of them might be of benefit in realizing the contributions of our paper.

Some recent theoretical progress has been made based on the projected subgradient method. In \cite{SMAP}, the inexact projections are adopted, but these projections require approaching the exact one in the course of the algorithm. Another approximate subgradient projection method is introduced in \cite{Kiwiel}. Rather than approximate projection, it considers approximate subgradient. The ZAP algorithm \cite{ZAP} is essentially a special case of the non-convex approach introduced in our paper. The literature \cite{l1ZAP} attempts to provide the convergence analysis of ZAP, yet the analysis is only for $\ell_1$-ZAP which uses the convex $\ell_1$ norm as the sparsity-inducing penalty. Despite this fact, it already contains some important ideas which are helpful in the theoretical analysis of our paper.

Since the introduction of the concept of weak convexity \cite{Janin}, a branch of researches has been focused on the duality and optimality conditions for weakly convex minimization problems \cite{Jeyakumar,DCoptimization,Wu}. These researches can be regarded as the extensions of those in convex optimization. They mainly consider the condition under which a point is the global minimizer of a weakly convex problem, which differs from the goal of our paper: providing convergence guarantees of an algorithm. In the area of sparse recovery, little attention has previously been paid to the concept of weak convexity. Our paper can be regarded as a pioneer work to introduce the concept of weak convexity to the field of compressive sensing and sparse recovery, and we believe that there is still much room for further research.

To verify the theoretical analysis in our paper, numerical simulations are implemented in the setting of random Gaussian sensing matrices. We have noticed that there is previous research characterizing the precise behavior of general penalization terms with Gaussian sensing matrices. One may read \cite{Bayati} for further reference.

\section{Main Contribution}\label{Sec_Main}

The main contributions of this paper are threefold. First, by combining the concept of sparseness measure with weak convexity, most commonly used sparsity-inducing penalties are characterized and some new results on the performance evaluation of $J$-minimization are derived. Second, a non-convex algorithm based on projected subgradient method is proposed to solve $J$-minimization with performance guarantees. Last but not the least, a uniform approximate projection is adopted in the proposed algorithm to save computational resources, and its performance guarantees as well as computational complexity analysis are provided. These contributions are demonstrated in the following subsections respectively.

\subsection{Performance Evaluation of $J$-minimization in the Noiseless Scenario}\label{subsec_perf}

Our work adopts weakly convex sparseness measure to constitute the sparsity-inducing penalty $J(\cdot)$ in (\ref{lFopt}). The definition of weakly convex sparseness measure is proposed as follows.

\begin{definition}\label{definition_weak_spar}
The weakly convex sparseness measure $F:\mathbb{R}\rightarrow\mathbb{R}$ satisfies
\begin{enumerate}
\item
$F(0)=0$, $F(\cdot)$ is even and not identically zero;
\item
$F(\cdot)$ is non-decreasing on $[0,+\infty)$;
\item
The function $t\mapsto F(t)/t$ is non-increasing on $(0,+\infty)$;
\item
$F(\cdot)$ is a weakly convex function on $[0,+\infty)$.
\end{enumerate}
\end{definition}

\begin{table}[t]
\renewcommand{\arraystretch}{1.5}
\caption{Weakly Convex Sparseness Measures with Parameter $\rho$\protect\\(Requirements: $0\le p<1$ and $\sigma>0$)}
\begin{center}
\begin{tabular}{ccc}
\toprule[1pt]
No. & $F(t)$ & $\rho$ \label{table constant}\\
\hline
1. & $|t|$ & 0\\
2. & $\frac{|t|}{(|t|+\sigma)^{1-p}}$ & $(p-1)\sigma^{p-2}$\\
3. & $1-{\rm e}^{-\sigma|t|}$ & $-\sigma^2/2$\\
4. & $\ln(1+\sigma|t|)$ & $-\sigma^2/2$\\
5. & ${\rm atan}(\sigma|t|)$ & $ -3\sqrt{3}\sigma^2/16$\\
6. & $(2\sigma |t|-{\sigma^2}t^2)\mathcal{X}_{|t|\le\frac{1}{\sigma}}+\mathcal{X}_{|t|>\frac{1}{\sigma}}$ & $-\sigma^2$\\
\bottomrule[1pt]
\end{tabular}
\end{center}
\end{table}

\begin{figure}[t]
\begin{center}
\includegraphics[width=4in]{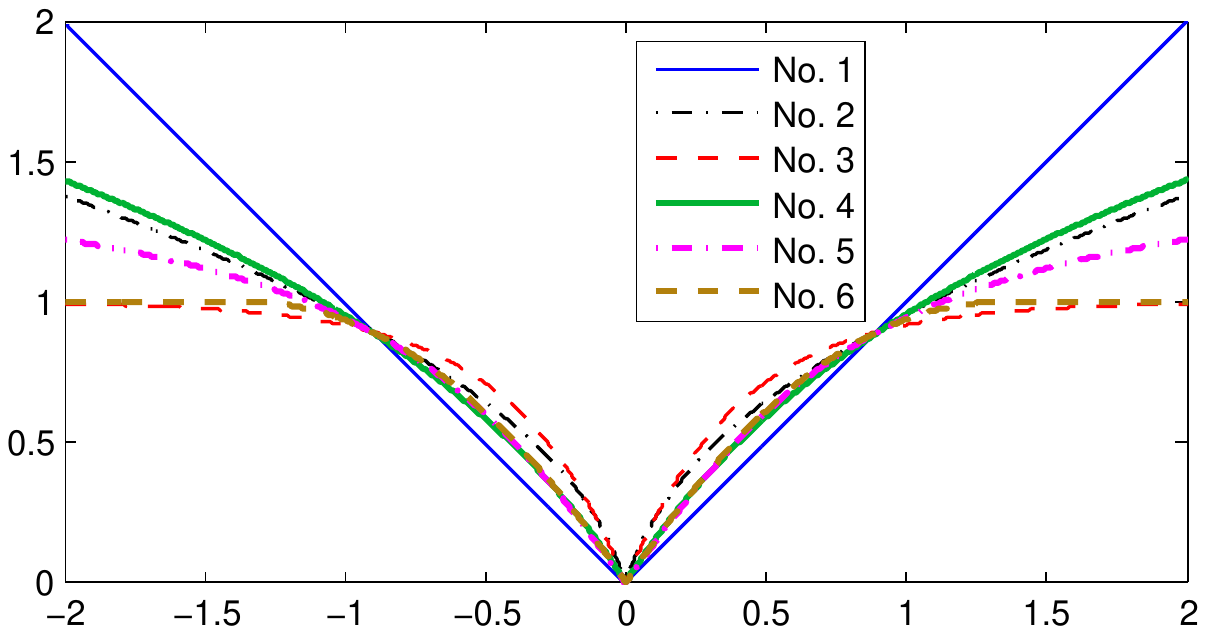}
\caption{The weakly convex sparseness measures listed in TABLE~\ref{table constant} are plotted. The parameter $p$ is set to $0.5$. The parameter $\sigma$ is set respectively so that they all contain the point $(0.9,0.9)$.}\label{picsparsemeasure}
\end{center}
\end{figure}

Definition~\ref{definition_weak_spar} is essentially a combination of the concepts of sparseness measure and weak convexity. Most commonly used non-convex penalties are formed by weakly convex sparseness measures. For instance, those penalties in \cite{RL1,ZAP,lqmin,Trzasko} are listed in TABLE~\ref{table constant} and plotted in Fig.~\ref{picsparsemeasure}, where $\mathcal{X}_P$ denotes the indicator function
\begin{align*}
\mathcal{X}_P=\left\{
\begin{array}{cl} 1 & P\textrm{ is true}; \\
0 & P\textrm{ is false}. \end{array}\right .
\end{align*}

It needs to be emphasized that the widely used $\ell_p$ ``norm'' $(0\le p<1)$ in the literatures of sparse recovery \cite{nonconvexopt,lqmin} does not belong to the class of sparsity-inducing penalties considered in this paper. This is due to the fact that the function
\begin{align}\label{lpfunction}
L_p(t)=|t|^p,\quad p\in[0,1)
\end{align}
goes against Definition~\ref{definition_weak_spar}.4), i.e., the requirement of weak convexity. However, approximations to (\ref{lpfunction}) are usually introduced to avoid infinite derivative around zero point and to improve the robustness. For example, in \cite{lqmin}, the function (\ref{lpfunction}) is approximated by
\begin{align*}
F(t)=\frac{|t|}{(|t|+\sigma)^{1-p}},\quad p\in[0,1),\sigma>0.
\end{align*}
This approximation satisfies Definition~\ref{definition_weak_spar}, and its parameter $\rho$ is shown in TABLE~\ref{table constant}. It hints that the requirement of weak convexity is reasonable and is an implicit assumption when robust algorithms or theoretical analysis is taken into consideration, which indicates the necessity of the introduction of weak convexity in this paper.

When the concept of sparseness measure meets weak convexity, some good properties show up.

\begin{lemma}\label{lemmapre}
The weakly convex sparseness measure $F(\cdot)$ satisfies the following properties:
\begin{enumerate}
\item
$F(\cdot)$ is continuous and there exists $\alpha>0$ such that $F(t)\le\alpha|t|$ holds for all $t\in\mathbb{R}$;
\item
For any constant $\beta>0$, $F(\beta t)$ is also a weakly convex sparseness measure, and its corresponding parameters are $\rho_{\beta}=\beta^2\rho$ and $\alpha_{\beta}=\beta\alpha$.
\end{enumerate}
\end{lemma}

\begin{proof}
The proof is postponed to Section~\ref{prooflemmapre}.
\end{proof}

Besides $\rho$, the parameter $\alpha$ also plays an important role in characterizing the non-convexity of sparsity-inducing penalty $J(\cdot)$. Recalling $J$-minimization (\ref{lFopt}), its performance remains the same for any positive scaled version of the penalty $J(\cdot)$. Since the parameters of $\beta F(t)$ are $\rho^{\beta}=\beta\rho$ and $\alpha^{\beta}=\beta\alpha$ for $\beta>0$, we let $-\rho/\alpha$ characterize the non-convexity, where $-\rho$ divided by $\alpha$ is to remove the scaling effect on the penalty. The No.~6 weakly convex sparseness measure in TABLE~\ref{table constant} is plotted in Fig.~\ref{picsnonconvexity} with the same $\alpha=2$ but different non-convexity. As can be seen, non-convexity can be regarded as a measure of how quickly the generalized gradient of $F(\cdot)$ decreases. Lemma~\ref{lemmapre}.2) implies that the non-convexity of $J(\beta{\bf x})$ is
\begin{align}\label{nonconvexity}
\frac{-\rho_{\beta}}{\alpha_{\beta}}=\beta\frac{-\rho}{\alpha}
\end{align}
for $\beta>0$. This reveals that by choosing an appropriate $\beta$, we can always generate a sparsity-inducing penalty with any desired non-convexity.

\begin{figure}[t]
\begin{center}
\includegraphics[width=4in]{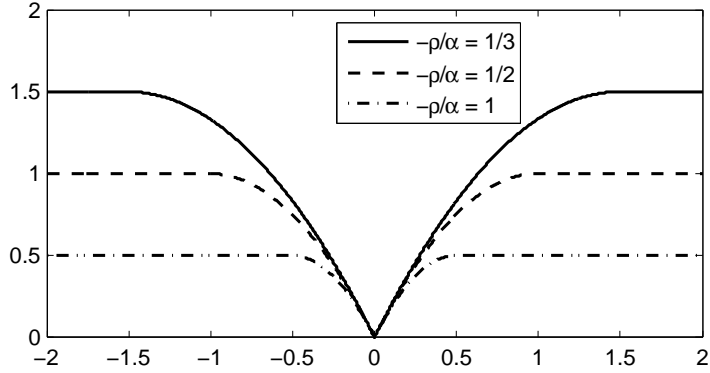}
\caption{The No.~6 weakly convex sparseness measure in TABLE~\ref{table constant} is plotted with different non-convexity. The parameter $\alpha$ is set to $2$.}\label{picsnonconvexity}
\end{center}
\end{figure}

The following theorem evaluates the performance of $J$-minimization for tuple $(J,{\bf A},{\bf x})$ under certain circumstances.

\begin{theorem}\label{theorem_l0}
Assume the tuple $({\bf A},K)$ satisfies $\gamma(\ell_0,{\bf A},K)<1$ and the vector ${\bf x}^*$ satisfies $\|{\bf x}^*\|_0\le K$. For any penalty $J(\cdot)$ formed by $F(\cdot)$ satisfying Definition~\ref{definition_weak_spar} and that $F(\cdot)$ is bounded, the global optimum $\hat{\bf x}^{\beta}$ of the problem
\begin{align}\label{opthl0}
\underset{\bf x}{\operatorname{argmin}}J(\beta{\bf x})\ \ \operatorname{subject\ to}\ \ {\bf Ax}={\bf Ax}^*
\end{align}
satisfies
\begin{align*}
\lim_{\beta\rightarrow+\infty}\|\hat{\bf x}^{\beta}-{\bf x}^*\|_2=0.
\end{align*}
\end{theorem}

\begin{proof}
The proof is postponed to Section~\ref{proofthl0}.
\end{proof}

Since the non-convexity of $J(\beta{\bf x})$ is (\ref{nonconvexity}), Theorem~\ref{theorem_l0} reveals that for a fixed sparse signal ${\bf x}^*$, the performance of $J$-minimization is close to that of $\ell_0$-minimization when the corresponding weakly convex sparseness measure is bounded and its non-convexity is large enough. One may notice that the condition in Theorem~\ref{theorem_l0} is $\gamma(\ell_0,{\bf A},K)<1$ rather than $\gamma(J,{\bf A},K)<1$ or $\gamma(\ell_1,{\bf A},K)<1$. As a matter of fact, $\gamma(\ell_0,{\bf A},K)<1$ is equivalent to the requirement of $M\ge2K+1$ and that any $2K$ column vectors of $\bf A$ are linearly independent, which is a much weaker condition than $\gamma(\ell_1,{\bf A},K)<1$. Therefore, for some $K$-sparse signals, they cannot be recovered by $\ell_1$-minimization, but can be recovered by $J$-minimization as shown in Theorem~\ref{theorem_l0}.

Recalling that the null space constant is a tight quantity for tuple $(J,{\bf A},K)$, a result on the performance of $J$-minimization is further derived from another perspective of view.

\begin{theorem}\label{theorem_NSP}
For any penalty $J(\cdot)$ formed by weakly convex sparseness measure $F(\cdot)$ satisfying Definition~\ref{definition_weak_spar}, the null space constant satisfies
\begin{align}
\gamma(J,{\bf A},K)=\gamma(\ell_1,{\bf A},K).
\end{align}
\end{theorem}

\begin{proof}
The proof is postponed to Section~\ref{proofthNSP}.
\end{proof}

According to Theorem~\ref{theorem_NSP}, for any tuple $({\bf A},K)$ and penalty $J(\cdot)$ formed by weakly convex sparseness measure, the performance of $J$-minimization is the same as that of $\ell_1$-minimization in the worst case sense. It needs to be noted that, although the performance comparison between $J$-minimization and $\ell_1$-minimization for any tuple $({\bf A},{\bf x})$ is still unclear in our work, some important related works have also run into the same situation. In \cite{nonconvexopt,Saab}, it is shown that for tuple $({\bf A},K)$, the condition under which $\ell_p$-minimization $(0<p<1)$ is guaranteed to find all $K$-sparse signals is weaker than that of $\ell_1$-minimization, and this is also the worst case analysis. We do believe that the performance comparison between the non-convex optimization and $\ell_1$-minimization for tuple $({\bf A},{\bf x})$ is worthy of further study, as it is the key point to all the literatures introducing non-convex techniques to sparse recovery \cite{FOCUSS,IRLS,RL1,SL0,DCAlgorithm,ISL0,ZAP}. So far, we speculate that
\begin{quote}
\emph{For tuple $({\bf A},{\bf x}^*)$, as $\beta$ increases from zero to positive infinity, the performance of (\ref{opthl0}) would gradually improve from $\ell_1$-minimization to some optimization problems with better performance, say $\ell_0$-minimization.}
\end{quote}
We leave this as a possible future work and it is readdressed in the conclusion of this paper.

\subsection{Projected Generalized Gradient Method in the Noisy Scenario}

Borrowing the idea of projected subgradient method, we propose a non-convex algorithm to solve the $J$-minimization problem. Mathematically, initialized as the pseudo-inverse solution ${\bf x}(0)={\bf A}^{\dagger}{\bf y}$ where ${\bf A}^{\dagger}={\bf A}^{\rm T}({\bf AA}^{\rm T})^{-1}$ denotes the pseudo-inverse matrix of $\bf A$, the iterative solution ${\bf x}(n)$ obeys
\begin{align}
\tilde{\bf x}(n+1)&={\bf x}(n)-\kappa\nabla J({\bf x}(n)),\label{itera1}\\
{\bf x}(n+1)&=\tilde{\bf x}(n+1)+{\bf A}^{\dagger}({\bf y}-{\bf A}\tilde{\bf x}(n+1)),\label{itera2}
\end{align}
where $\kappa>0$ denotes the step size and $\nabla J({\bf x})$ is a column vector whose $i$th element is $f(x_i)\in\partial F(x_i)$ which denotes the generalized gradient set of $F(\cdot)$ at $x_i$. Since the generalized gradient is adopted to update the iterative solutions, this method is termed projected generalized gradient (PGG) method in this paper. The procedure of PGG is described in TABLE~\ref{PGGmethod}. The algorithm stops when the iteration number exceeds a certain bound.

\begin{table}[t]
\renewcommand{\arraystretch}{1.5}
\caption{The Procedure of the PGG Method}
\begin{center}
\begin{tabular}{l}
\toprule[1pt]
{\bf Input:} \hspace{0.5em} ${\bf A}$, ${\bf y}$, step size $\kappa$, stopping criterion; \label{PGGmethod}\\
{\bf Output:} \hspace{0.5em} ${\bf x}(n)$.\\
\hline
{\bf Initialization:} \hspace{0.5em} Calculate ${\bf A}^{\dagger}$, ${\bf x}(0)={\bf A}^{\dagger}{\bf y}$, $n=0$;\\
{\bf Repeat:}\\
\hspace{1.5em} Generalized gradient step:\\
\hspace{3.5em} Update iterative solution by (\ref{itera1});\\
\hspace{1.5em} Projection step:\\
\hspace{3.5em} Update iterative solution by (\ref{itera2});\\
\hspace{1.5em} Iteration number increases by one:\\
\hspace{3.5em} $n=n+1$;\\
{\bf Until:} \hspace{0.5em} Stopping criterion satisfied;\\
\bottomrule[1pt]
\end{tabular}
\end{center}
\end{table}

In the remaining content of this subsection, we consider the performance of PGG in the noisy scenario ${\bf y}={\bf Ax}^*+{\bf e}$ where ${\bf x}^* $ is the $K$-sparse signal to be recovered and $\bf e$ is the additive noise to the measurement vector. Define $\sigma_{\min}({\bf A})$ as the smallest nonzero singular value of ${\bf A}$. The following theorem reveals the performance of PGG in the noisy scenario.

\begin{theorem}\label{theorem_PGG}
(Performance of PGG). For any tuple $(J,{\bf A},K)$ with $J(\cdot)$ formed by weakly convex sparseness measure $F(\cdot)$ and $\gamma(J,{\bf A},K)<1$, and for any positive constant $M_0$, if the non-convexity of $J(\cdot)$ satisfies
\begin{align}\label{theoremfor4}
\frac{-\rho}{\alpha}\le\frac{1}{M_0}\frac{1-\gamma(J,{\bf A},K)}{5+3\gamma(J,{\bf A},K)},
\end{align}
the recovered solution $\hat{\bf x}$ by PGG satisfies
\begin{align}
\|\hat{\bf x}-{\bf x}^*\|_2\le \frac{4\alpha^2N}{C_1}\kappa+8C_2\|{\bf e}\|_2
\end{align}
provided that $\|{\bf x}^*\|_0\le K$ and $\|{\bf x}(0)-{\bf x}^*\|_2\le M_0$, where
\begin{align}
C_1=&\frac{F(M_0)}{M_0}\frac{1-\gamma(J,{\bf A},K)}{1+\gamma(J,{\bf A},K)},\label{lemmafor4}\\
C_2=&\frac{\alpha\sqrt{N}+C_1}{C_1\sigma_{\min}({\bf A})}.\label{lemmafor5}
\end{align}
\end{theorem}

\begin{proof}
Theorem~\ref{theorem_PGG} can be directly derived from Lemma~\ref{maintheorem} and Lemma~\ref{theorem_new} in Section~\ref{Sec_Theo}.
\end{proof}

According to Theorem~\ref{theorem_PGG}, under some certain conditions, if the non-convexity of the penalty is below a threshold (which is in inverse proportion to the distance between the initial solution and the sparse signal), the recovered solution of PGG will get into the $(O(\kappa)+O(\|{\bf e}\|_2))$-neighborhood of ${\bf x}^*$. By choosing sufficiently small step size $\kappa$, the influence of the $O(\kappa)$ term can be omitted, and the PGG method returns a stably recovered solution. If $\rho=0$, $J(\cdot)$ is just a scaled version of the $\ell_1$ norm, and the condition (\ref{theoremfor4}) always holds for all $M_0>0$. Therefore, no constraint needs to be imposed on the distance between the initial solution and the sparse signal. This is consistent in the fact that $\ell_1$-minimization is convex and the initial solution can be arbitrary. In addition, larger non-convexity of the penalty induces smaller $M_0$, i.e., stronger constraint on the distance between the initial solution and the sparse signal, which is also an intuitive result.

\subsection{Extension and Discussion}

The initialization and the projection step of the PGG method involves the pseudo-inverse matrix ${\bf A}^{\dagger}$, whose exact calculation may be computationally intractable or even impossible because of its large scale in practical applications. To reduce the computational burden, a uniform approximate pseudo-inverse matrix of $\bf A$ is adopted. This method is termed approximate PGG (APGG) method. According to Appendix~\ref{app_cal} which introduces approximate calculation of the pseudo-inverse matrix, we use ${\bf A}^{\rm T}{\bf B}$ to denote the approximation of ${\bf A}^\dagger$. To characterize the approximate precision of the pseudo-inverse matrix, define
\begin{align*}
\|{\bf I}-{\bf A}{\bf A}^{\rm T}{\bf B}\|_2\le\zeta
\end{align*}
where $\|\cdot\|_2$ denotes the spectral norm of the matrix, and we assume $\zeta<1$ throughout this paper. Similar to Theorem~\ref{theorem_PGG}, the following theorem shows the performance of APGG in the noisy scenario.

\begin{theorem}\label{theorem_APGG}
(Performance of APGG). For any tuple $(J,{\bf A},K)$ with $J(\cdot)$ formed by weakly convex sparseness measure $F(\cdot)$ and $\gamma(J,{\bf A},K)<1$, and for any positive constant $M_0$, if the non-convexity of $J(\cdot)$ satisfies (\ref{theoremfor4}) and the approximate pseudo-inverse matrix ${\bf A}^{\rm T}{\bf B}$ satisfies $\zeta<1$, the recovered solution $\hat{\bf x}$ by APGG satisfies
\begin{align}\label{theoremfor5}
\|\hat{\bf x}-{\bf x}^*\|_2\le 2C_3\kappa+2C_4\|{\bf e}\|_2
\end{align}
provided that $\|{\bf x}^*\|_0\le K$ and $\|{\bf x}(0)-{\bf x}^*\|_2\le M_0$, where
\begin{align}
C_3&=\max\left\{2C_2C_5,\frac{2d\alpha^2N}{C_1}+C_6\right\},\label{lemmafor6}\\
C_4&=\max\left\{2C_2,C_7\right\},\label{lemmafor7}\\
C_5&=2\frac{\zeta\alpha\sqrt{N}\|{\bf A}\|_2}{1-\zeta},\label{lemmafor8}\\
C_6&=\frac{2\|{\bf B}\|_2C_5}{C_1}\left(2(1+\zeta)\alpha\sqrt{N}\|{\bf A}\|_2+(3+\zeta)C_5\right),\label{lemmafor9}\\
C_7&=\frac{4\|{\bf B}\|_2}{C_1}\left(\alpha\sqrt{N}\|{\bf A}\|_2+C_5\right),\label{lemmafor10}
\end{align}
$d=\|{\bf I}-{\bf A}^{\rm T}{\bf BA}\|_2^2$, and $C_1$ and $C_2$ are respectively specified as (\ref{lemmafor4}) and (\ref{lemmafor5}).
\end{theorem}

\begin{proof}
Theorem~\ref{theorem_APGG} can be directly derived from Lemma~\ref{lemma_APGG} and Lemma~\ref{theorem_new} in Section~\ref{Sec_Theo}.
\end{proof}

Similar to Theorem~\ref{theorem_PGG}, Theorem~\ref{theorem_APGG} also reveals that under some certain conditions, if the non-convexity of the penalty is below a threshold, the recovered solution of APGG will get into the $(O(\kappa)+O(\|{\bf e}\|_2))$-neighborhood of ${\bf x}^*$. This result is interesting since the influence of the approximate projection is only reflected on the coefficients instead of an additional error term. In the noiseless scenario with sufficiently small step size $\kappa$, the sparse signal ${\bf x}^*$ can be recovered with any given precision, even when a uniform approximate projection is adopted in this method.

By far, only the case of strictly sparse signal is analyzed and discussed. For compressible signal ${\bf x}^*$, assume $\|{\bf x}^*-{\bf x}_T^*\|_2\le\tau$. It is easily calculated that
\begin{align*}
{\bf y}={\bf Ax}^*+{\bf e}={\bf Ax}_T^*+({\bf e}+{\bf A}({{\bf x}^*-{\bf x}_T^*}))
\end{align*}
and
\begin{align*}
\|{\bf e}+{\bf A}({{\bf x}^*-{\bf x}_T^*})\|_2\le\|{\bf e}\|_2+\|{\bf A}\|_2\tau.
\end{align*}
According to Theorem~\ref{theorem_APGG}, the recovered solution $\hat{\bf x}$ of APGG will get into the $(2C_3\kappa+2C_4(\|{\bf e}\|_2+\|{\bf A}\|_2\tau))$-neighborhood of ${\bf x}_T^*$. Since ${\bf x}_T^*$ lies in the $\tau$-neighborhood of ${\bf x}^*$, the distance between $\hat{\bf x}$ and ${\bf x}^*$ will be no more than
\begin{align*}
2C_3\kappa+2C_4\|{\bf e}\|_2+(2C_4\|{\bf A}\|_2+1)\tau.
\end{align*}
This reflects the performance degradation due to the noise and non-sparsity of the original signal.

To end up this section, we talk about the computational complexity of the APGG method. The following Theorem~\ref{theorem_iternum} reveals how many iterations are needed for APGG to derive the solution with desired accuracy.

\begin{theorem}\label{theorem_iternum}
For any tuple $(J,{\bf A},K)$ with $J(\cdot)$ formed by weakly convex sparseness measure $F(\cdot)$ and $\gamma(J,{\bf A},K)<1$, positive constant $M_0$, vector ${\bf x}^*$ with $\|{\bf x}^*\|_0\le K$, and ${\bf A}^{\rm T}{\bf B}$ as an approximate pseudo-inverse matrix with $\zeta<1$, if the initial solution of APGG satisfies $\|{\bf x}(0)-{\bf x}^*\|_2\le M_0$ and the non-convexity of $J(\cdot)$ satisfies (\ref{theoremfor4}), then in at most
\begin{align*}
\frac{4C_3M_0}{d\alpha^2N\kappa}
\end{align*}
iterations, the recovered solution by APGG satisfies (\ref{theoremfor5}), where $C_3$ and $C_4$ are respectively specified as (\ref{lemmafor6}) and (\ref{lemmafor7}) and $d=\|{\bf I}-{\bf A}^{\rm T}{\bf BA}\|_2^2$.
\end{theorem}

\begin{proof}
The proof is postponed to Section~\ref{proofthiternum}.
\end{proof}

For calculating the approximate pseudo-inverse matrix of ${\bf A}$, the computational complexity of the method introduced in Appendix~\ref{app_cal} would be $O(MN)$ (if the initialization is adopted) or $O(M^2N)$ (if the method iterates for at least once). According to \eqref{lemmafor6}, it can be derived that $C_3$ is $O(N)$, therefore Theorem~\ref{theorem_iternum} reveals that the number of iterations needed is $O(\kappa^{-1})$. As for each iteration of APGG, the computational complexity is $O(MN)$. Overall, the computational complexity of APGG is at most $O(M^2N)+O(MN\kappa^{-1})$.

\section{Theoretical Analysis}\label{Sec_Theo}

This section mainly aims to establish theoretical supports for the results in Section~\ref{Sec_Main}. To begin with, some additional properties of weakly convex sparseness measure $F(\cdot)$ are revealed in the following lemma. Let $\partial F(0)=\{0\}$.

\begin{lemma}\label{lemmarest}
The weakly convex sparseness measure $F(\cdot)$ satisfies the following properties:
\begin{enumerate}
\item
For all $t_1,t_2\in\mathbb{R}$, $F(t_1+t_2)\le F(t_1)+F(t_2)$;
\item
For all $t\in(0,+\infty)$ and $f(t)\in\partial F(t)$, $f(t)\ge0$;
\item
For all $t\in\mathbb{R}$ and $f(t)\in\partial F(t)$, $|f(t)|\le\alpha$;
\item
For all $t_1,t_2\in\mathbb{R}$ and $f(t_1)\in\partial F(t_1)$, it holds that
\begin{align}\label{lemmapre3}
(t_1-t_2)f(t_1)\ge F(t_1)-F(t_2)+\rho(t_1-t_2)^2;
\end{align}
\item
For all $t\in\mathbb{R}$, $F(t)-\alpha|t|-\rho t^2\ge0$.
\end{enumerate}
\end{lemma}

\begin{proof}
The proof is postponed to Section~\ref{prooflemmarest}.
\end{proof}

Based on the definitions of weakly convex sparseness measure and null space constant with their properties, a lemma is established for preparation as follows.

\begin{lemma}\label{lemmanew}
For any tuple $(J,{\bf A},K)$ with $J(\cdot)$ formed by weakly convex sparseness measure $F(\cdot)$ and $\gamma(J,{\bf A},K)<1$, and for any positive constant $M_0$, the inequality
\begin{align}\label{lemmafor3}
J({\bf x})-J({\bf x}^*)\ge C_1\left(\|{\bf x}-{\bf x}^*\|_2-C_2\|{\bf A}({\bf x-x}^*)\|_2\right)
\end{align}
holds for all vectors ${\bf x}^*$ and $\bf x$ satisfying $\|{\bf x}^*\|_0\le K$ and $\|{\bf x-x}^*\|_2\le M_0$, where $C_1$ and $C_2$ are respectively specified as (\ref{lemmafor4}) and (\ref{lemmafor5}).
\end{lemma}

\begin{proof}
The proof is postponed to Section~\ref{prooflemmanew}.
\end{proof}

The following corollary can be immediately derived from Lemma~\ref{lemmanew}.

\begin{corollary}\label{coroconstant}
For any tuple $(J,{\bf A},K)$ with $J(\cdot)$ formed by weakly convex sparseness measure $F(\cdot)$ and $\gamma(J,{\bf A},K)<1$, and for any positive constant $M_0$, the inequality
\begin{align}\label{lemmaform3}
J({\bf x})-J({\bf x}^*)\ge \frac{C_1}{2}\|{\bf x}-{\bf x}^*\|_2
\end{align}
holds for all vectors ${\bf x}^*$ and $\bf x$ satisfying $\|{\bf x}^*\|_0\le K$, $\|{\bf x-x}^*\|_2\le M_0$, and $\|{\bf x-x}^*\|_2\ge 2C_2\|{\bf A}({\bf x-x}^*)\|_2$, where $C_1$ and $C_2$ are specified as (\ref{lemmafor4}) and (\ref{lemmafor5}), respectively.
\end{corollary}

The inequality (\ref{lemmaform3}) is somewhat similar to the concept of Lipschitz continuity, but with the difference that the inequality sign is reversed. According to (\ref{lemmaform3}), if the gap between $J({\bf x})$ and $J({\bf x}^*)$ is small, $\bf x$ would not be far away from the sparse vector ${\bf x}^*$. The following Lemma~\ref{theoremlocal} demonstrates the main result on the local minima of $J$-minimization.

\begin{lemma}\label{theoremlocal}
For any tuple $(J,{\bf A},K)$ with $J(\cdot)$ formed by weakly convex sparseness measure $F(\cdot)$ and $\gamma(J,{\bf A},K)<1$, and for any positive constant $M_0$, the inequality
\begin{align}
({\bf x}-{\bf x}^*)^{\rm T}\nabla J({\bf x})>0
\end{align}
holds for all vectors ${\bf x}^*$ and $\bf x$ satisfying $\|{\bf x}^*\|_0\le K$,
\begin{align}\label{theoremfor1}
\|{\bf x-x}^*\|_2\le \min\left\{M_0,\frac{C_1}{-4\rho}\right\},
\end{align}
and $\|{\bf x-x}^*\|_2\ge 2C_2\|{\bf A}({\bf x-x}^*)\|_2$, where $C_1$ and $C_2$ are specified as (\ref{lemmafor4}) and (\ref{lemmafor5}), respectively.
\end{lemma}

\begin{proof}
The proof is postponed to Section~\ref{prooftheoremlocal}.
\end{proof}

Lemma~\ref{theoremlocal} demonstrates the distribution of the local minima of $J$-minimization. As is revealed, for any local minimum $\bf x$ in the area of (\ref{theoremfor1}), it also satisfies
\begin{align*}
\|{\bf x-x}^*\|_2\le2C_2\|{\bf A}({\bf x-x}^*)\|_2.
\end{align*}
Therefore, Lemma~\ref{theoremlocal} implies that there is no local minimum in the corresponding annulus. Intuitively, recalling that ${\bf A}({\bf x}(n)-{\bf x}^*)={\bf e}$ for the PGG method, if the initial solution satisfies (\ref{theoremfor1}), the recovered solution is stable against the noise. The following Lemma~\ref{maintheorem} demonstrates the detailed convergence property of the PGG method in one iteration. For simplicity, let $\bf x$ and ${\bf x}^+$ represent ${\bf x}(n)$ and ${\bf x}(n+1)$, respectively.

\begin{lemma}\label{maintheorem}
For any tuple $(J,{\bf A},K)$ with $J(\cdot)$ formed by weakly convex sparseness measure $F(\cdot)$ and $\gamma(J,{\bf A},K)<1$, positive constant $M_0$, and vector ${\bf x}^*$ with $\|{\bf x}^*\|_0\le K$, if the previous iterative solution ${\bf x}$ of the PGG method satisfies (\ref{theoremfor1}) and
\begin{align}\label{theoremfor2}
\|{\bf x}-{\bf x}^*\|_2\ge \frac{2\mu\alpha^2N}{C_1}\kappa+4C_2\|{\bf e}\|_2,
\end{align}
where $\mu>1$ and $C_1$ and $C_2$ are respectively specified as (\ref{lemmafor4}) and (\ref{lemmafor5}), the next iterative solution ${\bf x}^+$ satisfies
\begin{align}
\|{\bf x}^+-{\bf x}^*\|_2^2\le\|{\bf x}-{\bf x}^*\|_2^2-(\mu-1)\alpha^2N\kappa^2.
\end{align}
\end{lemma}

\begin{proof}
The proof is postponed to Section~\ref{proofmaintheorem}.
\end{proof}

According to Lemma~\ref{maintheorem}, if the iterative solution ${\bf x}(n)$ lies within a neighborhood of the sparse signal ${\bf x}^*$ as (\ref{theoremfor1}), as long as the distance between ${\bf x}(n)$ and ${\bf x}^*$ is larger than a quantity linear in both the step size $\kappa$ and the noise term $\|{\bf e}\|_2$, the next iterative solution ${\bf x}(n+1)$ will definitely get closer to ${\bf x}^*$, and the distance reduction is at least $(\mu-1)\alpha^2N\kappa^2$. Therefore, in finite iterations, the iterative solution ${\bf x}(n)$ will get into the $(O(\kappa)+O(\|{\bf e}\|_2))$-neighborhood of ${\bf x}^*$.

To ensure that the PGG method converges, we require the sufficient condition (\ref{theoremfor1}) satisfied for the initial solution. We can simply choose parameters such that
\begin{align}\label{rhocons}
M_0=\|{\bf x}(0)-{\bf x}^*\|_2\le\frac{C_1}{-4\rho}.
\end{align}
The following lemma reveals that penalties with small non-convexity will result in (\ref{rhocons}).

\begin{lemma}\label{theorem_new}
For any tuple $(J,{\bf A},K)$ with $J(\cdot)$ formed by weakly convex sparseness measure $F(\cdot)$ and $\gamma(J,{\bf A},K)<1$, and for any positive constant $M_0$, the constraint (\ref{rhocons}) holds if the non-convexity of $J(\cdot)$ satisfies (\ref{theoremfor4}).
\end{lemma}

\begin{proof}
The proof is postponed to Section~\ref{prooftheoremnew}.
\end{proof}

Next we consider the performance of the APGG method. Since ${\bf A}^{\rm T}{\bf B}$ is adopted as the approximation of ${\bf A}^{\dagger}$, the iterative solution of APGG no longer satisfies ${\bf Ax}(n)={\bf y}$. The following lemma gives the bound of $\|{\bf A}({\bf x}(n)-{\bf x}^*)\|_2$.

\begin{lemma}\label{theorem_space}
The iterative solution ${\bf x}(n)$ of the APGG method satisfies
\begin{align}
\|{\bf A}({\bf x}(n)-{\bf x}^*)\|_2\le \|{\bf y}\|_2\zeta^{n+1}+\frac{1}{2}C_5\kappa+\|{\bf e}\|_2,
\end{align}
where $C_5$ is specified as (\ref{lemmafor8}).
\end{lemma}

\begin{proof}
The proof is postponed to Section~\ref{prooftheoremspace}.
\end{proof}

According to Lemma~\ref{theorem_space}, if the accurate pseudo-inverse matrix is applied, i.e., $\zeta=0$, the result is consistent in the scenario with accurate projection. For any fixed approximate precision $\zeta\in(0,1)$, as $n$ approaches infinity and the step size $\kappa$ is sufficiently small, the result reveals that the performance degradation caused by the approximate projection can be omitted. For the convenience of theoretical analysis, define a constant $N_{\kappa}$ such that for all $n\ge N_{\kappa}$,
\begin{align*}
\|{\bf A}({\bf x}(n)-{\bf x}^*)\|_2\le C_5\kappa+\|{\bf e}\|_2.
\end{align*}

Since Lemma~\ref{lemmanew}, Corollary~\ref{coroconstant}, Lemma~\ref{theoremlocal}, and Lemma~\ref{theorem_new} are independent of specific algorithms, they still hold for the APGG method. The following lemma demonstrates the convergence property of APGG in one iteration, which is a counterpart of Lemma~\ref{maintheorem}.

\begin{lemma}\label{lemma_APGG}
For any tuple $(J,{\bf A},K)$ with $J(\cdot)$ formed by weakly convex sparseness measure $F(\cdot)$ and $\gamma(J,{\bf A},K)<1$, positive constant $M_0$, vector ${\bf x}^*$ with $\|{\bf x}^*\|_0\le K$, and ${\bf A}^{\rm T}{\bf B}$ as an approximate pseudo-inverse matrix with $\zeta<1$, if the previous iterative solution ${\bf x}$ of the APGG method satisfies (\ref{theoremfor1}) and
\begin{align}\label{itfirstnoise}
\|{\bf x}-{\bf x}^*\|_2\ge\mu C_3\kappa+C_4\|{\bf e}\|_2,
\end{align}
where $\mu>1$ and $C_3$ and $C_4$ are respectively specified as (\ref{lemmafor6}) and (\ref{lemmafor7}), the next iterative solution ${\bf x}^+$ satisfies
\begin{align}\label{lemmafor11}
\|{\bf x}^+-{\bf x}^*\|_2^2\le\|{\bf x}-{\bf x}^*\|_2^2-(\mu-1)d\alpha^2N\kappa^2,
\end{align}
where $d=\|{\bf I}-{\bf A}^{\rm T}{\bf BA}\|_2^2$.
\end{lemma}

\begin{proof}
The proof is postponed to Section~\ref{prooftheoremAPGG}.
\end{proof}

\section{Numerical Simulation}\label{Sec_Simu}

In this section, several simulations are implemented to test the recovery performance of the (A)PGG method, and to verify the theoretical analysis. The sensing matrix $\bf A$ is of size $M=200$ and $N=1000$, whose entries are independently and identically distributed Gaussian with zero mean and variance $1/M$. The locations of the nonzero entries of the sparse signal ${\bf x}^*$ are randomly chosen among all possible choices, and these nonzero entries satisfy Gaussian distribution or symmetric Bernoulli distribution with zero mean. The sparse signal is finally normalized to have unit $\ell_2$ norm. In all simulations, the approximate ${\bf A}^{\dagger}$ is calculated using the method introduced in Appendix~\ref{app_cal}.

\begin{figure}[t]
\begin{center}
\includegraphics[width=4in]{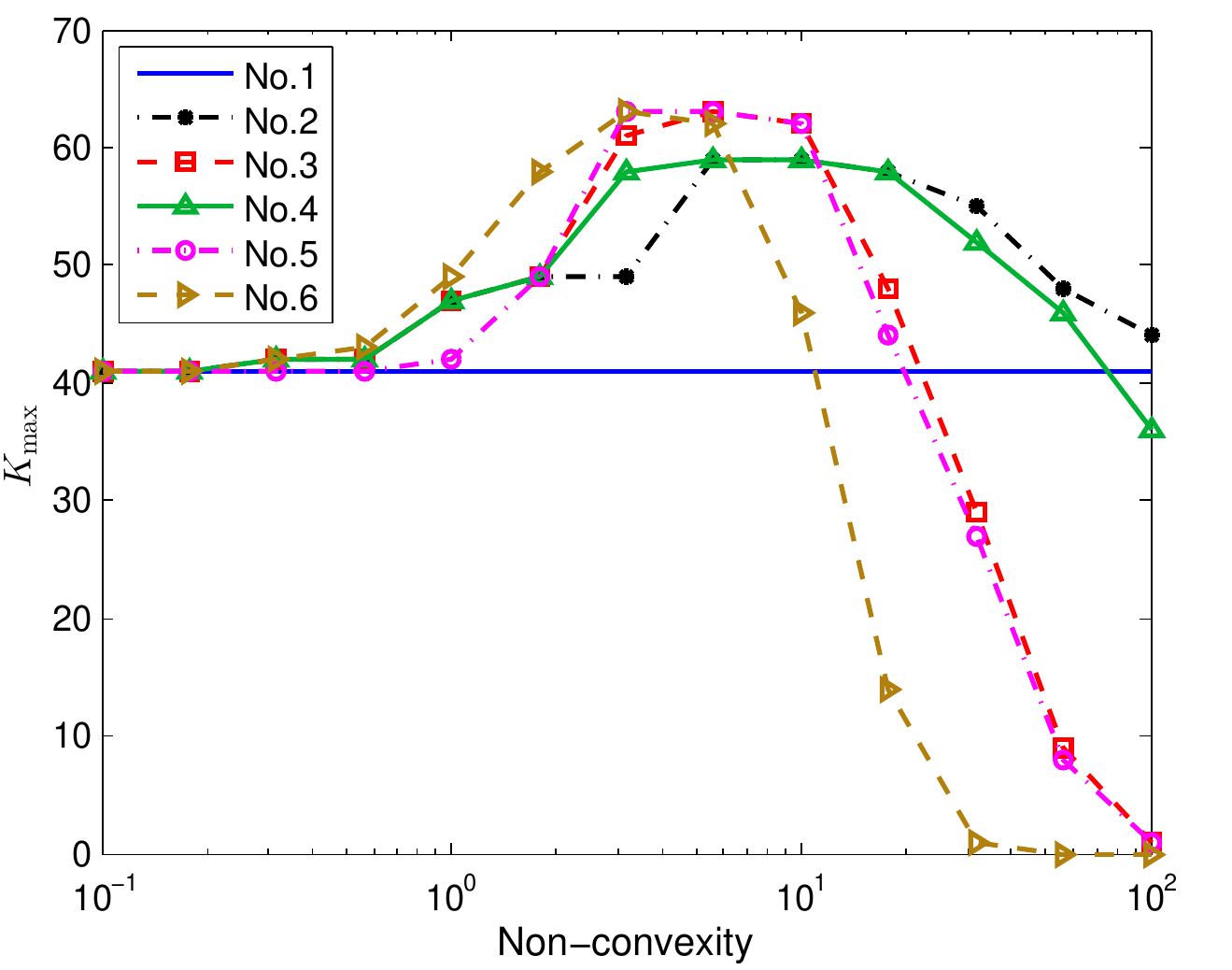}
\caption{The figure shows the recovery performance of the PGG method with different sparsity-inducing penalties and different choices of non-convexity when the nonzero entries of the sparse signal satisfy Gaussian distribution. The corresponding sparseness measures are from TABLE~\ref{table constant}. The problem dimensions are $M=200$ and $N=1000$, and $K_{\max}$ is the largest integer which guarantees $100\%$ successful recovery.}\label{maxkrho}
\end{center}
\end{figure}

The first experiment tests the recovery performance of the PGG method in the noiseless scenario with different sparsity-inducing penalties and different choices of non-convexity. The penalties are formed by sparseness measures in TABLE~\ref{table constant}. The parameter $p=0.5$ and $\sigma$ is set to have desired non-convexity. The No.~1 corresponds to the $\ell_1$ penalty, which is tested in the same parameter settings as a benchmark. The penalties are scaled so that the parameter $\alpha=1$. For each penalty with some certain non-convexity, the sparsity level $K$ varies from $1$ to $100$ with increment of one. The step size $\kappa$ is set to $1\times 10^{-5}$. If the recovery SNR (RSNR) is higher than $40$dB, this recovery is regarded as a success. The simulation is repeated $100$ times to calculate the successful recovery probability versus sparsity $K$. Then the crucial sparsity $K_{\max}$, which is the largest integer which guarantees $100\%$ successful recovery, is recorded. The results when the nonzero entries of the sparse signal satisfy Gaussian distribution and Bernoulli distribution are presented in Fig.~\ref{maxkrho} and Fig.~\ref{maxkrho1}, respectively. As can be seen from the results, as the non-convexity of the sparsity-inducing penalty increases, the performance of PGG improves at first, and degenerates when the non-convexity continues to grow. When the non-convexity approaches zero, the performances of these penalties are close to that of the $\ell_1$ penalty. The results support the speculation in the end of Section~\ref{subsec_perf} that as the non-convexity increases, the performance of $J$-minimization improves, and verify Theorem~\ref{theorem_PGG} that the non-convexity should be smaller than a threshold to guarantee the convergence of PGG.

\begin{figure}[t]
\begin{center}
\includegraphics[width=4in]{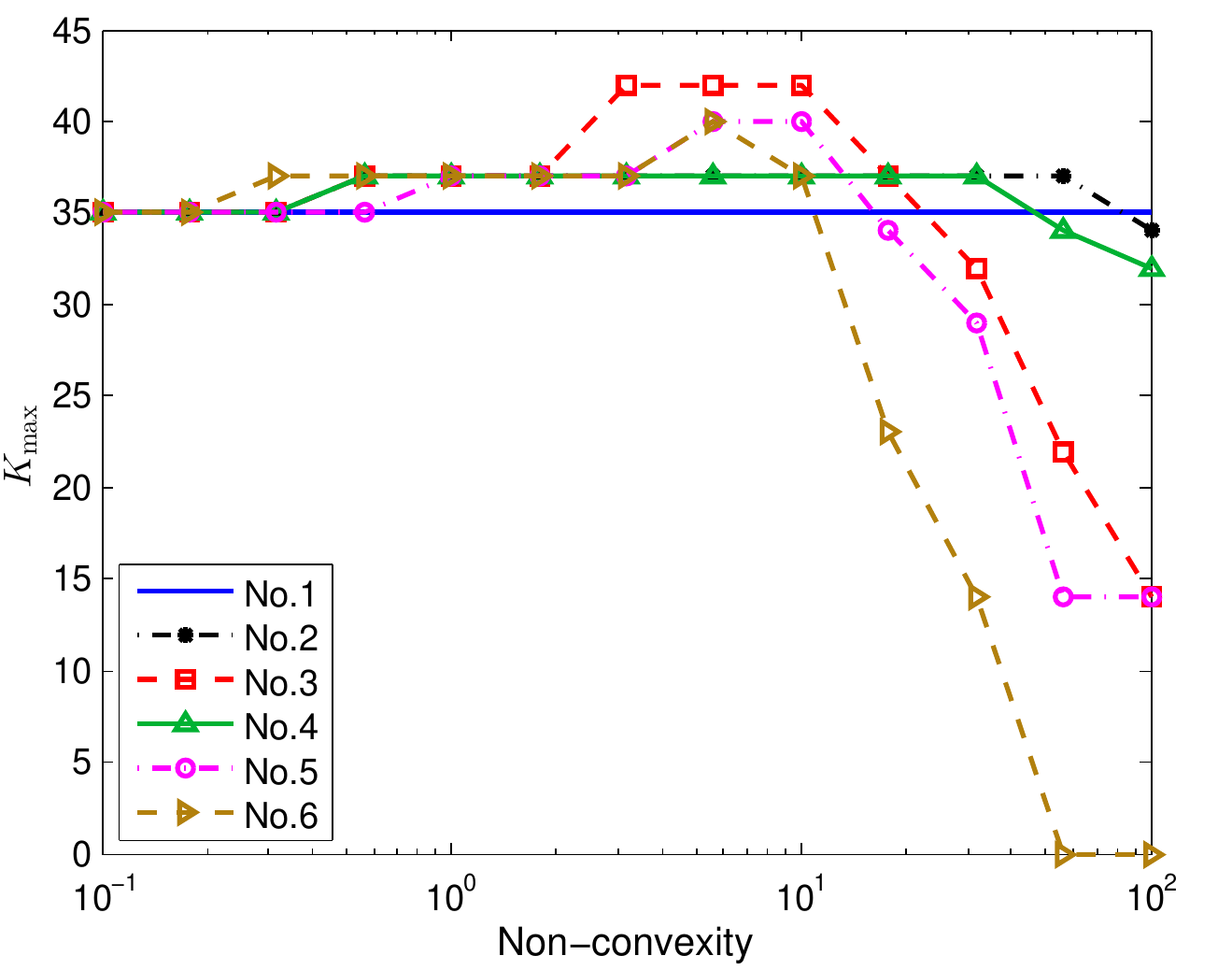}
\caption{The figure shows the recovery performance of the PGG method with different sparsity-inducing penalties and different choices of non-convexity when the nonzero entries of the sparse signal satisfy Bernoulli distribution. The corresponding sparseness measures are from TABLE~\ref{table constant}. The problem dimensions are $M=200$ and $N=1000$, and $K_{\max}$ is the largest integer which guarantees $100\%$ successful recovery.}\label{maxkrho1}
\end{center}
\end{figure}

In the second experiment, the recovery performance of (A)PGG is compared in the noiseless scenario with some typical sparse recovery algorithms, including orthogonal matching pursuit (OMP) \cite{OMP}, the solution to $\ell_1$-minimization \cite{cvx_1}, reweighted $\ell_1$ minimization \cite{RL1}, ISL0 \cite{ISL0}, and IRLS \cite{IRLS}. In the simulation $K$ varies from $20$ to $100$. The (A)PGG method adopts the No.~6 sparseness measure in TABLE~\ref{table constant} with non-convexity as $10^{0.75}$, and the penalty is scaled so that $\alpha=1$. The step size is set to $1\times 10^{-5}$. The iteration number for calculating inexact pseudo-inverse matrices is $0$ and the average approximate precision $\zeta=0.91$. The simulation is repeated $500$ times to calculate the successful recovery probability versus sparsity $K$. The simulation results when the nonzero entries of the sparse signal satisfy Gaussian distribution and Bernoulli distribution are demonstrated in Fig.~\ref{probability} and Fig.~\ref{probability1}, respectively. As can be seen, for both distributions, IRLS, PGG, and APGG guarantee successful recovery for larger sparsity $K$ than the other references. It also reveals that in the noiseless scenario with sufficiently small step size, the approximate projection has little influence on the recovery performance of APGG.

\begin{figure}[t]
\begin{center}
\includegraphics[width=4in]{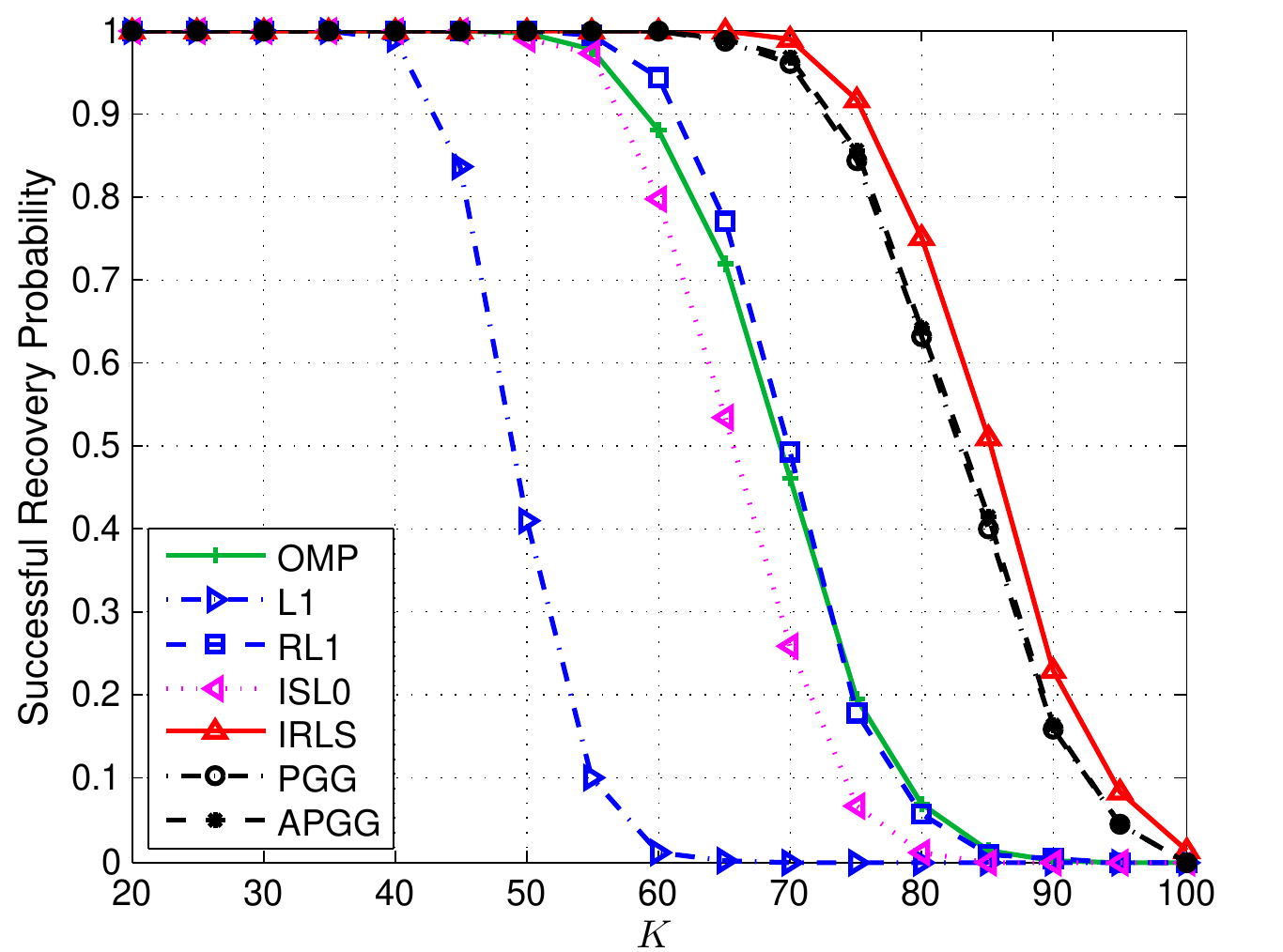}
\caption{The figure compares the successful recovery probability of different algorithms versus sparsity $K$ with $M=200$ and $N=1000$ when the nonzero entries of the sparse signal satisfy Gaussian distribution. The approximate precision of approximate ${\bf A}^{\dagger}$ is $\zeta=0.91$.}\label{probability}
\end{center}
\end{figure}

\begin{figure}[t]
\begin{center}
\includegraphics[width=4in]{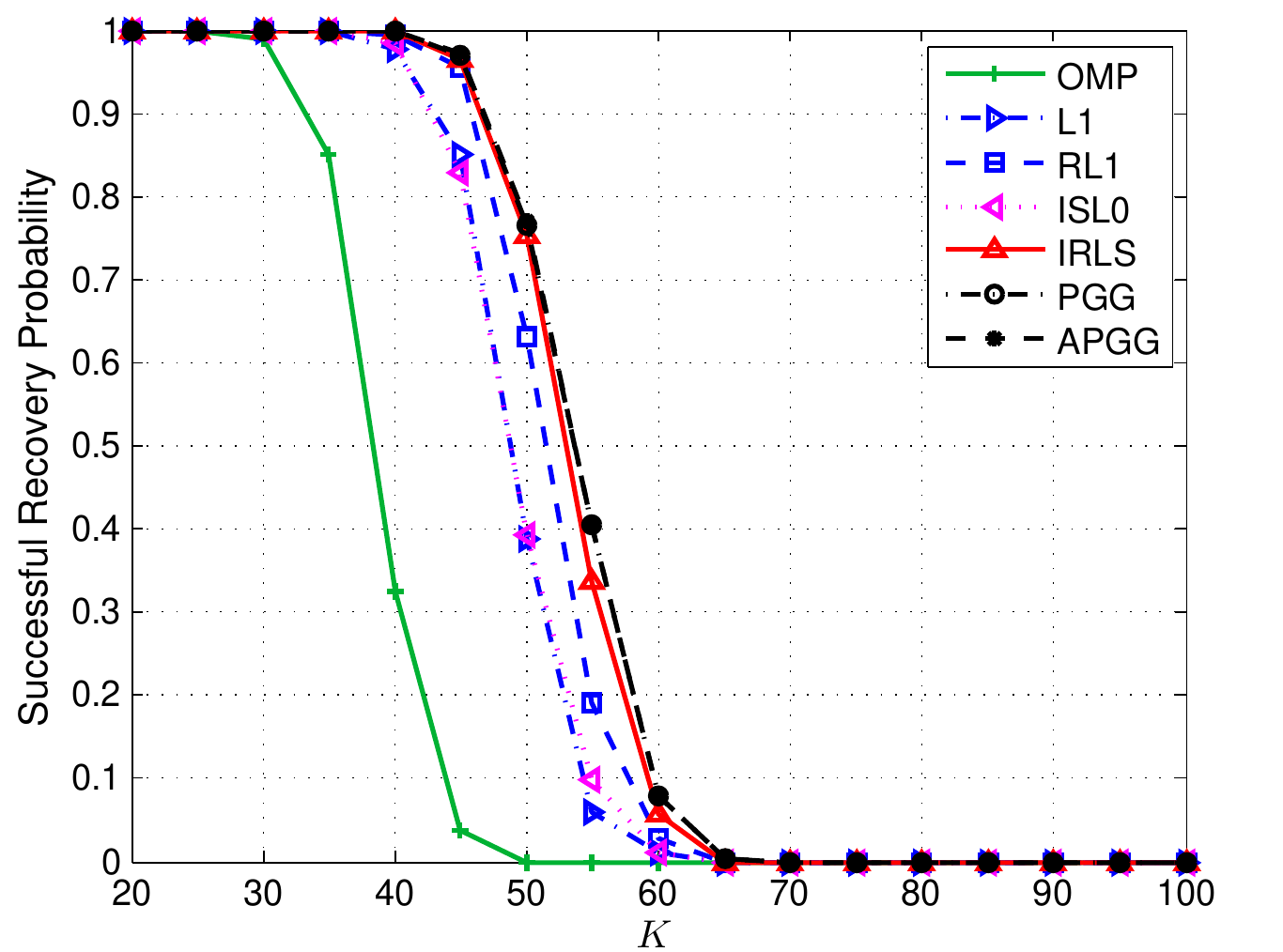}
\caption{The figure compares the successful recovery probability of different algorithms versus sparsity $K$ with $M=200$ and $N=1000$ when the nonzero entries of the sparse signal satisfy Bernoulli distribution. The approximate precision of approximate ${\bf A}^{\dagger}$ is $\zeta=0.91$.}\label{probability1}
\end{center}
\end{figure}

In the last experiment, the recovery precisions of the (A)PGG method are simulated under different settings of step size and measurement noise. In the simulation, the nonzero entries of the sparse signal satisfy Gaussian distribution and the sparsity level $K=30$. The same sparseness measure as that in the previous experiment is adopted, and the iteration number for calculating approximate ${\bf A}^{\dagger}$ is $4$ such that $\zeta=0.22$. The simulation is repeated $500$ times to calculate the $95\%$ confidence interval of RSNR and the average RSNR (which is defined as the mean relative root squared error in dB), and the results are shown in Fig.~\ref{SNRvskappa}. As can be seen, there is almost no difference between the performance of PGG and that of APGG. In the noisy scenario, the RSNR is dependent on both the step size and the measurement SNR (MSNR). For fixed MSNR, as the step size decreases, the RSNR improves at first, and remains the same when the step size is sufficiently small. Larger MSNR results in larger RSNR limit. In the noiseless scenario, the RSNR improves as the step size decreases, and it can be arbitrarily large by adopting sufficiently small step size. These results are accordant with Theorem~\ref{theorem_PGG} and Theorem~\ref{theorem_APGG}, which implies that the recovery error is linear in both the step size and the noise term.

\begin{figure}[t]
\begin{center}
\includegraphics[width=4in]{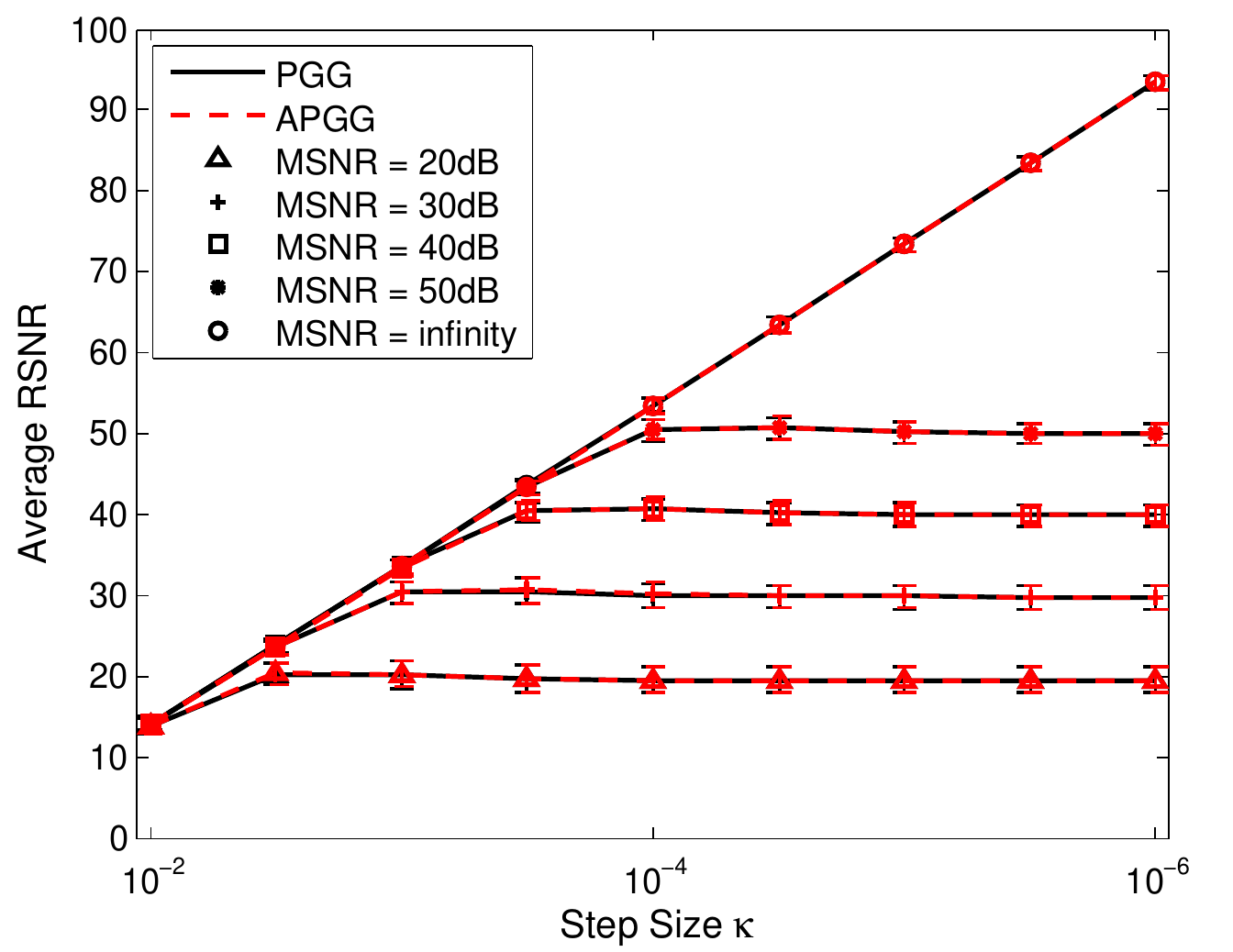}
\caption{The figure demonstrates the recovery precisions of the (A)PGG method with their $95\%$ confidence intervals under different step sizes and MSNRs with $M=200$, $N=1000$, and $K=30$ when the nonzero entries of the sparse signal satisfy Gaussian distribution. The approximate precision of approximate ${\bf A}^{\dagger}$ is $\zeta=0.22$.}\label{SNRvskappa}
\end{center}
\end{figure}

\section{Proof}\label{Sec_Proof}

\subsection{Proof of Lemma~\ref{lemmapre}}\label{prooflemmapre}

\begin{proof}

1) The continuity of $F(\cdot)$ can be easily checked by Proposition~\ref{proposition_1} and the continuity of convex functions. As for the inequality, we only need to consider the case of $t>0$. Since $F(t)/t$ is non-increasing on $(0,+\infty)$ and
\begin{align*}
\lim_{t\rightarrow0_+}\frac{F(t)}{t}=\lim_{t\rightarrow0_+}\left(\frac{H(t)}{t}+\rho t\right)=\lim_{t\rightarrow0_+}\frac{H(t)-H(0)}{t-0}\triangleq\alpha
\end{align*}
is a finite quantity, it holds that for all $t>0$, $F(t)/t\le\alpha$.

2) It is easy to check that $F(\beta t)$ satisfies Definition~\ref{definition_weak_spar}.1)-3). Since $F(\beta t)=H(\beta t)+\beta^2\rho t^2$ and $H(\beta t)$ is convex, $F(\beta t)$ satisfies Definition~\ref{definition_weak_spar}.4) with parameter $\rho_{\beta}=\beta^2\rho$. In addition, since $F(\beta t)/t=H(\beta t)/t+\beta^2\rho t$, the same argument as the proof of Lemma~\ref{lemmapre}.1) implies that $\alpha_{\beta}=\beta\alpha$.
\end{proof}

\subsection{Proof of Theorem~\ref{theorem_l0}}\label{proofthl0}

\begin{proof}
According to the definition of null space constant, $\gamma(\ell_0,{\bf A},K)<1$ implies that for any nonzero vector ${\bf z}\in\mathcal{N}({\bf A})$, $\bf z$ has at least $(2K+1)$ nonzero entries, and any $2K$ column vectors of ${\bf A}$ are linearly independent. Since $F(\cdot)$ is non-decreasing and bounded on $[0,+\infty)$, without loss of generality, we assume $\lim_{t\rightarrow+\infty}F(t)=C>0$.

For any $\varepsilon>0$, define
\begin{align*}
\delta=\frac{\varepsilon}{\sqrt{N}(D\|{\bf A}\|_2+1)}>0
\end{align*}
where $D^{-1}$ is the smallest singular value of all $2K$ column submatrices of $\bf A$ ($D^{-1}$ is nonzero since any $2K$ column vectors of ${\bf A}$ are linearly independent). Since $F(\cdot)$ is non-decreasing on $[0,+\infty)$, these exists $\beta_0>0$ such that for all $\beta>\beta_0$ and for all $t>\delta$, $F(\beta t)>\frac{K}{K+1}C$.

First we prove that for all $\beta>\beta_0$, $\hat{\bf x}^{\beta}$ has at most $K$ entries with absolute value no less than $\delta$. This is due to the fact that (define $I_{\beta}$ as the set of index $i$ satisfying $|\hat{x}^{\beta}_i|\ge\delta$)
\begin{align*}
KC\ge J(\beta{\bf x}^*)\ge J(\beta\hat{\bf x}^{\beta})\ge\sum_{i\in I_{\beta}}F(\beta\hat{x}^{\beta}_i)>\frac{K}{K+1}C\cdot\#I_{\beta}
\end{align*}
which implies $\#I_{\beta}\le K$. Together with $K$-sparse signal ${\bf x}^*$, at most $2K$ entries of $\hat{\bf x}^{\beta}-{\bf x}^*$ are with absolute value no less than $\delta$.

Now we prove that for all $\beta>\beta_0$, $\|\hat{\bf x}^{\beta}-{\bf x}^*\|_2\le\varepsilon$. Define ${\bf z}^{\beta}=\hat{\bf x}^{\beta}-{\bf x}^*$ and $I^{\beta}$ as the set of index $i$ satisfying $|{z}^{\beta}_i|\ge\delta$, then as has been proved, $\#I^{\beta}\le2K$. On the one hand,
\begin{align*}
\|{\bf z}_{(I^{\beta})^c}^{\beta}\|_2\le\sqrt{N}\delta.
\end{align*}
On the other hand, since ${\bf A}{\bf z}^{\beta}=\bf 0$,
\begin{align*}
\|{\bf z}_{I^{\beta}}^{\beta}\|_2\le D\|{\bf A}{\bf z}_{I^{\beta}}^{\beta}\|_2=D\|{\bf A}{\bf z}_{(I^{\beta})^c}^{\beta}\|_2\le D\|{\bf A}\|_2\sqrt{N}\delta.
\end{align*}
Therefore,
\begin{align*}
\|{\bf z}^{\beta}\|_2\le\|{\bf z}_{I^{\beta}}^{\beta}\|_2+\|{\bf z}_{(I^{\beta})^c}^{\beta}\|_2\le(D\|{\bf A}\|_2+1)\sqrt{N}\delta=\varepsilon
\end{align*}

To sum up, we have proved that for any $\varepsilon>0$, there exists $\beta_0>0$ such that for all $\beta>\beta_0$, $\|\hat{\bf x}^{\beta}-{\bf x}^*\|_2\le\varepsilon$. This directly leads to Theorem~\ref{theorem_l0}.
\end{proof}

\subsection{Proof of Theorem~\ref{theorem_NSP}}\label{proofthNSP}

\begin{proof}
Define a class of penalties $J_{\beta}({\bf x})=J(\beta{\bf x})$ for $\beta>0$. We first prove that for all $\beta>0$, $\gamma(J,{\bf A},K)=\gamma(J_{\beta},{\bf A},K)$. This can be easily proved from the definition of the null space constant and the fact that for all $\beta>0$, $\beta{\bf z}\in\mathcal{N}({\bf A})$ is equivalent to ${\bf z}\in\mathcal{N}({\bf A})$.

Now we prove $\gamma(J,{\bf A},K)=\gamma(\ell_1,{\bf A},K)$. If not, according to Proposition~\ref{proposition_NSP}.3), there exists $\delta>0$ such that for all $\beta>0$,
\begin{align}\label{app_equ1}
\gamma(J_{\beta},{\bf A},K)\le\gamma(\ell_1,{\bf A},K)-3\delta.
\end{align}
According to the definition of the null space constant, there exist ${\bf z}\in\mathcal{N}({\bf A})$ and set $S$ with $\#S\le K$ such that
\begin{align}\label{app_equ2}
\|{\bf z}_S\|_1/\|{\bf z}_{S^c}\|_1\ge\gamma(\ell_1,{\bf A},K)-\delta.
\end{align}
In addition, since for fixed ${\bf z}$ and $S$,
\begin{align*}
\lim_{\beta\rightarrow0_+}J(\beta{\bf z}_S)/J(\beta{\bf z}_{S^c})=\|{\bf z}_S\|_1/\|{\bf z}_{S^c}\|_1,
\end{align*}
there exists $\beta_0>0$ such that for all $0<\beta\le\beta_0$,
\begin{align}\label{app_equ3}
J(\beta{\bf z}_S)/J(\beta{\bf z}_{S^c})\ge\|{\bf z}_S\|_1/\|{\bf z}_{S^c}\|_1-\delta.
\end{align}
Combining (\ref{app_equ2}) with (\ref{app_equ3}), it can be derived that
\begin{align}
J(\beta{\bf z}_S)/J(\beta{\bf z}_{S^c})\ge\gamma(\ell_1,{\bf A},K)-2\delta
\end{align}
holds for all $0<\beta\le\beta_0$, which contradicts (\ref{app_equ1}).
\end{proof}

\subsection{Proof of Lemma~\ref{lemmarest}}\label{prooflemmarest}

\begin{proof}

1) Consider the non-trivial scenario where $t_1$ and $t_2$ are both nonzero. Since $F(t)/t$ is non-increasing on $(0,+\infty)$, it is easily checked that
\begin{equation*}
\begin{aligned}
F(t_1)=F(|t_1|)&\ge \left(|t_1|F(|t_1|+|t_2|)\right)/(|t_1|+|t_2|);\\
F(t_2)=F(|t_2|)&\ge \left(|t_2|F(|t_1|+|t_2|)\right)/(|t_1|+|t_2|).
\end{aligned}
\end{equation*}
Summing these two inequalities, together with the non-decreasing property of $F(\cdot)$ on $[0,+\infty)$, it holds that
\begin{align*}
F(t_1)+F(t_2)&\ge F(|t_1|+|t_2|)\ge F(|t_1+t_2|)=F(t_1+t_2).
\end{align*}

2) Since $F(\cdot)$ is non-decreasing on $[0,+\infty)$, the directional derivative
\begin{align*}
D_F(t,-1)=\lim_{\theta\rightarrow0_+}(F(t-\theta)-F(t))/\theta\le0
\end{align*}
holds for all $t>0$. Therefore, the definition of the generalized gradient set (\ref{definition_gen}) implies that for all $f(t)\in\partial F(t)$, $f(t)\ge0$.

3) It is easy to check that $F(\cdot)$ is also weakly convex on $(-\infty,0]$ with parameter $\rho$ and that for all $t\in\mathbb{R}$, $\partial F(-t)=-\partial F(t)$. Therefore we only need to consider the case of $t>0$. Due to the non-increasing property of $F(t)/t$, it can be verified that $(F(t+\theta)-F(t))/\theta\le F(t)/t$ holds for all $\theta>0$. Therefore the definition of the generalized gradient implies
\begin{align*}
0\le f(t)\le\lim_{\theta\rightarrow0_+}(F(t+\theta)-F(t))/\theta\le F(t)/t\le\alpha.
\end{align*}

4) First, if $(t_1,t_2)$ satisfies the inequality (\ref{lemmapre3}), it is easy to check that $(-t_1,-t_2)$ also satisfies it, therefore we only need to consider the scenario that $t_1\ge 0$.

If $t_1=0$, the result is obvious since $\rho\le0$. If $t_1>0$ and $t_2\ge0$, according to Proposition~\ref{proposition_3} and the fact that $F(\cdot)$ is weakly convex with parameter $\rho$ on $[0,+\infty)$, the inequality (\ref{lemmapre3}) is still obvious. If $t_1>0$ and $t_2<0$, then $-t_2>0$. Since $f(t_1)\ge 0$, it can be derived that
\begin{align*}
(t_1-t_2)f(t_1)&\ge F(t_1)-F(-t_2)+\rho(t_1+t_2)^2\\
&\ge F(t_1)-F(t_2)+\rho(t_1-t_2)^2.
\end{align*}
To sum up, the inequality (\ref{lemmapre3}) is proved.

5) Assume $F(t)=H(t)+\rho t^2$ and decompose $H(\cdot)$ by $H(t)=\alpha|t|+G(t)$. Since $H(\cdot)$ is convex, according to the definition of $\alpha$,  $G(t)\ge0$.
\end{proof}

\subsection{Proof of Lemma~\ref{lemmanew}}\label{prooflemmanew}

\begin{proof}
Define ${\bf u}={\bf x}-{\bf x}^*$ and decompose ${\bf u}$ by ${\bf u}={\bf z}+{\bf z}^{\bot}$, where ${\bf z}\in \mathcal{N}({\bf A})$ and ${\bf z}^{\bot}\in\mathcal{N}({\bf A})^{\bot}$, which denotes the orthogonal complement of $\mathcal{N}({\bf A})$. Therefore ${\bf Az}^{\bot}={\bf Au}$. Since $\sigma_{\min}({\bf A})$ is the smallest nonzero singular value of ${\bf A}$,
\begin{align}\label{lemmaproofz1}
\|{\bf z}^{\bot}\|_2\le \|{\bf Au}\|_2/\sigma_{\min}({\bf A}).
\end{align}
Supposing that ${\bf x}^*$ is supported on $T$ and according to Lemma~\ref{lemmarest}.1), it can be derived that
\begin{align}\label{lemmaproof2}
J({\bf x})-J({\bf x}^*)&=J({\bf x}^*+{\bf u}_T)-J({\bf x}^*)+J({\bf u}_{T^c})\nonumber\\
&\ge J({\bf u}_{T^c})-J({\bf u}_T).
\end{align}
By the decomposition of ${\bf u}$, it can be further derived from Lemma~\ref{lemmarest}.1) that
\begin{align}\label{lemmaproofknex1}
J({\bf x})-J({\bf x}^*)\ge J({\bf z}_{T^c})-J({\bf z}_T)-J({\bf z}^{\bot}).
\end{align}

On the one hand, according to the definition of null space constant,
\begin{align}\label{prooflemma2}
J({\bf z}_{T^c})-J({\bf z}_T)\ge\frac{1-\gamma(J,{\bf A},K)}{1+\gamma(J,{\bf A},K)}J({\bf z}).
\end{align}
On the other hand, according to Lemma~\ref{lemmapre}.1) and (\ref{lemmaproofz1}),
\begin{align}\label{lemmaproofknex3}
J({\bf z}^{\bot})\le \alpha\|{\bf z}^{\bot}\|_1\le\alpha\sqrt{N}\|{\bf Au}\|_2/\sigma_{\min}({\bf A}).
\end{align}
Since for $1\le i\le N$, $|z_i|\le\|{\bf z}\|_2\le \|{\bf u}\|_2\le M_0$, it can be calculated that
\begin{align}\label{prooflemma3}
J({\bf z})\ge F(M_0)\|{\bf z}\|_1/M_0\ge F(M_0)\|{\bf z}\|_2/M_0,
\end{align}
where the first inequality is due to Definition~\ref{definition_weak_spar}.3). Therefore (\ref{lemmaproofknex1}), (\ref{prooflemma2}), (\ref{lemmaproofknex3}), and (\ref{prooflemma3}) imply
\begin{align}\label{lemmaproofz2}
J({\bf x})-J({\bf x}^*)\ge C_1\|{\bf z}\|_2-\alpha\sqrt{N}\|{\bf Au}\|_2/\sigma_{\min}({\bf A}).
\end{align}
Since $\|{\bf u}\|_2\le\|{\bf z}\|_2+\|{\bf z}^{\bot}\|_2$, according to (\ref{lemmaproofz1}), (\ref{lemmafor3}) can be directly derived.
\end{proof}

\subsection{Proof of Lemma~\ref{theoremlocal}}\label{prooftheoremlocal}

\begin{proof}
According to Lemma~\ref{lemmarest}.4), it can be derived that
\begin{align}\label{proofknexfor4}
({\bf x}-{\bf x}^*)^{\rm T}\nabla J({\bf x})&\ge J({\bf x})-J({\bf x}^*)+\rho\|{\bf x-x}^*\|_2^2.
\end{align}
Since $\|{\bf x-x}^*\|_2\le\frac{C_1}{-4\rho}$, Corollary~\ref{coroconstant} and (\ref{proofknexfor4}) imply
\begin{align}\label{proofknexfor6}
({\bf x}-{\bf x}^*)^{\rm T}\nabla J({\bf x})\ge C_1\|{\bf x-x}^*\|_2/4,
\end{align}
which completes the proof.
\end{proof}

\subsection{Proof of Lemma~\ref{maintheorem}}\label{proofmaintheorem}

\begin{proof}
Define ${\bf u}={\bf x}-{\bf x}^*$ and ${\bf u}^+={\bf x}^+-{\bf x}^*$. According to the procedure of PGG, it can be derived that ${\bf u}^+={\bf u}-\kappa({\bf I}-{\bf A}^{\dagger}{\bf A})\nabla J({\bf x})$, which further implies
\begin{equation}\label{proofknexfor5}
\begin{aligned}
\|{\bf u}^+\|_2^2=&\|{\bf u}\|_2^2+\kappa^2\|({\bf I}-{\bf A}^{\dagger}{\bf A})\nabla J({\bf x})\|_2^2\\
&-2\kappa{\bf u}^{\rm T}({\bf I}-{\bf A}^{\dagger}{\bf A})\nabla J({\bf x}).
\end{aligned}
\end{equation}
According to Lemma~\ref{lemmarest}.3), the second item on the right side of (\ref{proofknexfor5}) can be bounded as
\begin{align*}
\|({\bf I}-{\bf A}^{\dagger}{\bf A})\nabla J({\bf x})\|_2^2\le\|\nabla J({\bf x})\|_2^2\le\alpha^2N.
\end{align*}
The third item on the right side of (\ref{proofknexfor5}) can be decomposed to
\begin{align*}
{\bf u}^{\rm T}({\bf I}-{\bf A}^{\dagger}{\bf A})\nabla J({\bf x})={\bf u}^{\rm T}\nabla J({\bf x})-{\bf u}^{\rm T}{\bf A}^{\dagger}{\bf A}\nabla J({\bf x}).
\end{align*}
On the one hand, according to the proof of Lemma~\ref{theoremlocal}, (\ref{proofknexfor6}) implies that
\begin{align*}
{\bf u}^{\rm T}\nabla J({\bf x})\ge C_1\|{\bf u}\|_2/4.
\end{align*}
On the other hand,
\begin{align*}
{\bf u}^{\rm T}{\bf A}^{\dagger}{\bf A}\nabla J({\bf x})\le\alpha\sqrt{N}\|{\bf A}{\bf u}\|_2/\sigma_{\min}({\bf A}).
\end{align*}
Substituting these inequalities into (\ref{proofknexfor5}) and according to (\ref{theoremfor2}), the right side of (\ref{proofknexfor5}) can be bounded as $\|{\bf u}\|_2^2-(\mu-1)\alpha^2N\kappa^2$, which arrives Lemma~\ref{maintheorem}.
\end{proof}

\subsection{Proof of Lemma~\ref{theorem_new}}\label{prooftheoremnew}

\begin{proof}
According to the definition of $C_1$ and Lemma~\ref{lemmarest}.5),
\begin{align}
\frac{C_1}{-4\rho}\ge\frac{\alpha M_0+\rho M_0^2}{-4\rho M_0}\frac{1-\gamma(J,{\bf A},K)}{1+\gamma(J,{\bf A},K)}.
\end{align}
Therefore, due to (\ref{theoremfor4}), the constraint (\ref{rhocons}) holds.
\end{proof}

\subsection{Proof of Lemma~\ref{theorem_space}}\label{prooftheoremspace}

\begin{proof}
First, we prove that
\begin{align}
\|{\bf y}-{\bf Ax}(n)\|_2\le\|{\bf y}\|_2\zeta^{n+1}+C_5\kappa/2.
\end{align}
For $n=0$, the initialization is ${\bf x}(0)={\bf A}^{\rm T}{\bf By}$, which satisfies
\begin{align}\label{initerror}
\|{\bf y}-{\bf A}{\bf x}(0)\|_2=\|{\bf y}-{\bf A}{\bf A}^{\rm T}{\bf By}\|_2\le\|{\bf y}\|_2\zeta.
\end{align}
For the $(n+1)$th iteration, the iterative solution obeys
\begin{align}\label{iteration}
{\bf x}(n+1)={\bf A}^{\rm T}{\bf B}{\bf y}+({\bf I}-{\bf A}^{\rm T}{\bf BA})({\bf x}(n)-\kappa\nabla J({\bf x}(n))),
\end{align}
which satisfies
\begin{equation*}
\begin{aligned}
&\|{\bf y}-{\bf A}{\bf x}(n+1)\|_2\\
=&\|({\bf I}-{\bf A}{\bf A}^{\rm T}{\bf B})({\bf y}-{\bf A}({\bf x}(n)-\kappa\nabla J({\bf x}(n))))\|_2\\
\le&\|{\bf y}-{\bf A}{\bf x}(n)\|_2\zeta+\alpha\sqrt{N}\|{\bf A}\|_2\kappa\zeta.
\end{aligned}
\end{equation*}
Together with (\ref{initerror}), it can be derived by recursion that
\begin{align}
\|{\bf y}-{\bf A}{\bf x}(n)\|_2&\le\|{\bf y}-{\bf A}{\bf x}(0)\|_2\zeta^n+\frac{\zeta\alpha\sqrt{N}\|{\bf A}\|_2}{1-\zeta}\cdot\kappa\nonumber\\
&\le\|{\bf y}\|_2\zeta^{n+1}+C_5\kappa/2,\label{prooflemma4}
\end{align}

Now we turn to the proof of Lemma~\ref{theorem_space}. Since ${\bf y=Ax}^*+{\bf e}$, it can be derived that
\begin{align*}
\|{\bf A}({\bf x}(n)-{\bf x}^*)\|_2&\le\|{\bf y-Ax}(n)\|_2+\|{\bf y-Ax}^*\|_2\\
&\le\|{\bf y}\|_2\zeta^{n+1}+C_5\kappa/2+\|{\bf e}\|_2,
\end{align*}
which completes the proof.
\end{proof}

\subsection{Proof of Lemma~\ref{lemma_APGG}}\label{prooftheoremAPGG}

\begin{proof}
Similar to the proof of Lemma~\ref{maintheorem}, define ${\bf u}={\bf x}-{\bf x}^*$ and ${\bf u}^+={\bf x}^+-{\bf x}^*$. According to (\ref{iteration}), it holds that ${\bf u}^+={\bf u}+{\bf A}^{\rm T}{\bf B}({\bf y}-{\bf A}{\bf x})-\kappa({\bf I}-{\bf A}^{\rm T}{\bf BA})\nabla J({\bf x})$, which further implies
\begin{equation}\label{proof3}
\begin{aligned}
\|{\bf u}^+\|_2^2=&\|{\bf u}\|_2^2+\|{\bf A}^{\rm T}{\bf B}({\bf y}-{\bf A}{\bf x})\|_2^2+2{\bf u}^{\rm T}{\bf A}^{\rm T}{\bf B}({\bf y}-{\bf A}{\bf x})\\
&+\kappa^2\|({\bf I}-{\bf A}^{\rm T}{\bf BA})\nabla J({\bf x})\|_2^2\\
&-2\kappa{\bf u}^{\rm T}({\bf I}-{\bf A}^{\rm T}{\bf BA})\nabla J({\bf x})\\
&-2\kappa({\bf y}-{\bf A}{\bf x})^{\rm T}{\bf B}^{\rm T}{\bf A}({\bf I}-{\bf A}^{\rm T}{\bf BA})\nabla J({\bf x}).
\end{aligned}
\end{equation}

According to (\ref{prooflemma4}) and $n\ge N_{\kappa}$, for the second item on the right side of (\ref{proof3}),
\begin{align*}
\|{\bf A}^{\rm T}{\bf B}({\bf y}-{\bf A}{\bf x})\|_2^2\le(1+\zeta)\|{\bf B}\|_2C^2_5\kappa^2.
\end{align*}
For the third item,
\begin{align*}
{\bf u}^{\rm T}{\bf A}^{\rm T}{\bf B}({\bf y}-{\bf A}{\bf x})\le\|{\bf B}\|_2C_5\kappa\left(C_5\kappa+\|{\bf e}\|_2\right).
\end{align*}
For the forth item,
\begin{align*}
\|({\bf I}-{\bf A}^{\rm T}{\bf BA})\nabla J({\bf x})\|_2^2\le\|{\bf I}-{\bf A}^{\rm T}{\bf BA}\|_2^2\alpha^2N=d\alpha^2N
\end{align*}
For the fifth item, it can be decomposed to
\begin{align*}
{\bf u}^{\rm T}({\bf I}-{\bf A}^{\rm T}{\bf BA})\nabla J({\bf x})={\bf u}^{\rm T}\nabla J({\bf x})-{\bf u}^{\rm T}{\bf A}^{\rm T}{\bf BA}\nabla J({\bf x}).
\end{align*}
According to the proof of Lemma~\ref{theoremlocal}, since $\|{\bf x-x}^*\|_2\ge 2C_2(C_5\kappa+\|{\bf e}\|_2)$, (\ref{proofknexfor6}) implies that
\begin{align*}
{\bf u}^{\rm T}\nabla J({\bf x})\ge C_1 \|{\bf u}\|_2/4,
\end{align*}
and
\begin{align*}
{\bf u}^{\rm T}{\bf A}^{\rm T}{\bf BA}\nabla J({\bf x})\le\alpha\sqrt{N}\|{\bf A}\|_2\|{\bf B}\|_2\left(C_5\kappa+\|{\bf e}\|_2\right).
\end{align*}
For the last item,
\begin{align*}
&({\bf y}-{\bf A}{\bf x})^{\rm T}{\bf B}^{\rm T}{\bf A}({\bf I}-{\bf A}^{\rm T}{\bf BA})\nabla J({\bf x})\\
\ge&-\alpha\sqrt{N}\|{\bf A}\|_2\|{\bf B}\|_2\zeta C_5\kappa.
\end{align*}

Together with the above inequalities, (\ref{proof3}) can be simplified to
\begin{equation}\label{proofthfor1}
\begin{aligned}
\|{\bf u}^+\|_2^2\le&\|{\bf u}\|_2^2+d\alpha^2N\kappa^2\\
&-\frac{C_1}{2}\left(\|{\bf u}\|_2-C_6\kappa-C_7\|{\bf e}\|_2\right)\kappa,
\end{aligned}
\end{equation}
where $C_6$ and $C_7$ are specified as (\ref{lemmafor9}) and (\ref{lemmafor10}), respectively. Therefore, under the assumption (\ref{itfirstnoise}), inequality (\ref{proofthfor1}) implies (\ref{lemmafor11}), which completes the proof.
\end{proof}

\subsection{Proof of Theorem~\ref{theorem_iternum}}\label{proofthiternum}

\begin{proof}
Assume that the iterative solution of APGG satisfies
\begin{align}\label{proofknexfor7}
\|{\bf x}-{\bf x}^*\|_2\ge 2C_3\kappa+2C_4\|{\bf e}\|_2.
\end{align}
Since Lemma~\ref{lemma_APGG} holds for any $\mu>1$, we choose
\begin{align}
\mu(n)=\frac{\|{\bf x}-{\bf x}^*\|_2-C_4\|{\bf e}\|_2}{C_3\kappa}>1,
\end{align}
and the next iterative solution satisfies
\begin{align}
&\|{\bf x}^+-{\bf x}^*\|_2^2\nonumber\\
\le&\|{\bf x}-{\bf x}^*\|_2^2-(\mu(n)-1)d\alpha^2N\kappa^2\nonumber\\
=&\|{\bf x}-{\bf x}^*\|_2^2-\frac{d\alpha^2N\kappa}{C_3}\left(\|{\bf x}-{\bf x}^*\|_2-C_3\kappa-C_4\|{\bf e}\|_2\right)\nonumber\\
\le&\left(\|{\bf x}-{\bf x}^*\|_2-\frac{d\alpha^2N\kappa}{4C_3}\right)^2,
\end{align}
where the last inequality can be derived from the assumption (\ref{proofknexfor7}). Therefore,
\begin{align}\label{proofknexfor8}
\|{\bf x}^+-{\bf x}^*\|_2\le \|{\bf x}-{\bf x}^*\|_2-\frac{d\alpha^2N\kappa}{4C_3},
\end{align}
i.e., the distance reduction is at least $\frac{d\alpha^2N\kappa}{4C_3}$. Since the initial solution satisfies $\|{\bf x}(0)-{\bf x}^*\|_2\le M_0$, in at most
\begin{align*}
\frac{M_0}{\frac{d\alpha^2N\kappa}{4C_3}}=\frac{4C_3M_0}{d\alpha^2N\kappa}
\end{align*}
iterations, the recovered solution by APGG satisfies (\ref{theoremfor5}).

It needs to be noted that, similar to the discussions in Section III-E of \cite{l1ZAP}, $\mu$ is just a parameter in the theoretical analysis, and the choice of $\mu$ would not influence the actual convergence of iterations of APGG. In other words, the inequality (\ref{proofknexfor8}) always holds as long as the assumption (\ref{proofknexfor7}) holds, and this fact is independent of the choice of $\mu$.
\end{proof}

\section{Conclusion}\label{Sec_Conc}

This paper considers the convergence guarantees of a non-convex approach for sparse recovery. A class of weakly convex sparseness measures is adopted to constitute the sparsity-inducing penalties. The convergence analysis of the (A)PGG method reveals that when the non-convexity of the penalty is below a threshold (which is in inverse proportion to the distance between the initial solution and the sparse signal), the recovery error is linear in both the step size and the noise term. As for the APGG method, the influence of the approximate projection is reflected in the coefficients instead of an additional error term. Therefore, in the noiseless scenario with sufficiently small step size, APGG returns a solution with any given precision. Simulation results verify the theoretical analysis in this paper, and the recovery performance of APGG is not much influenced by the approximate projection.

There are several future directions to be explored. The first direction is to study the performance of $J$-minimization for tuple $({\bf A},{\bf x}^*)$. In this paper we mainly utilize the null space constant to characterize its performance, and it is only tight for tuple $({\bf A},K)$. For a fixed sparse signal ${\bf x}^*$, as the non-convexity $-\rho/\alpha$ increases, the performance of $J$-minimization should be different, as is revealed in Theorem~\ref{theorem_l0} and Fig.~\ref{maxkrho}-\ref{maxkrho1}. The second possible direction is to improve the performance of sparse recovery by solving a sequence of optimization problems with different choices of non-convexity. The major concern would be the selection rules of the sequence of non-convexity such that the recovered solution for the previous non-convexity would lie in the convergence neighborhood for the next non-convexity.

\appendices

\section{Approximate Calculation of ${\bf A}^{\dagger}$}\label{app_cal}

The methods of computing ${\bf A}^\dag$ have been developed to a mature technology. They are roughly classified into two categories: direct methods \cite{ShinozakiD} and iterative methods \cite{ShinozakiI}. Direct methods are mainly based on matrix decompositions, such as QR decomposition \cite{ShinozakiD} and singular value decomposition \cite{SVD1,SVD2}. Iterative methods, on the other hand, derive the pseudo-inverse matrix iteratively. To develop more accurate solutions, they cost more computational resources. Therefore, the iterative methods are preferred if approximate pseudo-inverse matrix can be applied to reduce the computational complexity.

A well-known iterative method introduced by Ben-Israel et al. \cite{IsraelI} is
\begin{equation*}
\begin{aligned}
{\bf Y}_0&=\varsigma {\bf A}^{\rm T},\\
{\bf Y}_k&={\bf Y}_{k-1}(2{\bf I}-{\bf AY}_{k-1})
\end{aligned}
\end{equation*}
with the parameter $\varsigma$ satisfying $0<\varsigma<2/\|{\bf AA}^{\rm T}\|_1$, where $\|\cdot\|_1$ denotes the maximum absolute column sum of the matrix. Simple calculation derives that
\begin{equation*}
\begin{aligned}
\|{\bf I}-{\bf AY}_0\|_2&=\|{\bf I}-\varsigma{\bf AA}^{\rm T}\|_2<1,\\
\|{\bf I}-{\bf AY}_k\|_2&\le\|{\bf I}-{\bf AY}_{k-1}\|_2^2\le\|{\bf I}-{\bf AY}_0\|_2^{2^k},
\end{aligned}
\end{equation*}
which means this method is quadratic convergence.

In this paper, it is assumed that the approximate pseudo-inverse matrix is of the form ${\bf A}^{\rm T}{\bf B}$, i.e., the transpose of $\bf A$ multiplied by a matrix ${\bf B}\in\mathbb{R}^{M\times M}$. ${\bf B}$ is considered as the approximation of $({\bf A}{\bf A}^{\rm T})^{-1}$. It is verified that most, if not all, iterative methods \cite{ShinozakiI,IsraelI,Genbook} satisfy this assumption.

%


\end{document}